\documentclass[12pt]{amsart}
\usepackage[margin=1in]{geometry}
\usepackage{amssymb,amsthm, times}
\usepackage{delarray,verbatim,bbm,enumerate}

\usepackage{ifpdf}
\usepackage{color}
\definecolor{webgreen}{rgb}{0,.5,0}
\definecolor{webbrown}{rgb}{.8,0,0}
\definecolor{emphcolor}{rgb}{0.95,0.95,0.95}

\usepackage{hyperref}
\hypersetup{%
          colorlinks=true,
          linkcolor=webbrown,
          filecolor=webbrown,
          citecolor=webgreen,
          breaklinks=true}
\ifpdf \hypersetup{pdftex,
            pdfstartview=FitH, 
            bookmarksopen=true,
            bookmarksnumbered=true
} \else \hypersetup{dvips} \fi

\DeclareMathAlphabet{\mathpzc}{OT1}{pzc}{m}{it}

\DeclareMathOperator{\hol}{C}

\numberwithin{equation}{section}
\DeclareMathOperator{\expo}{e}

\theoremstyle{plain}
  \newtheorem{theorem}{Theorem}[section]
  
  \newtheorem{proposition}[theorem]{Proposition}
  \newtheorem{lemma}[theorem]{Lemma}
  
\theoremstyle{remark}
  \newtheorem{remark}[theorem]{Remark}

\newcommand {\R}{\mathbb{R}}

\newcommand {\E}{\mathbb{E}}

\newcommand{\diff}{{\rm d}}
\newcommand{\sell}{{\rm s}}
\newcommand{\barrier}{{\rm b}}
\newcommand{\cont}{{\rm c}}

\newcommand{\blue}{\textcolor[rgb]{0.00,0.00,1.00}}

\newcommand{\eqdef}{\raisebox{0.4pt}{\ensuremath{:}}\hspace*{-1mm}=}
\newcommand{\defeq}{=\hspace*{-1mm}\raisebox{0.4pt}{\ensuremath{:}}}

\title{Periodic strategies in optimal execution with multiplicative price impact}

\author[D. Hern\'andez]{Daniel Hern\'andez-Hern\'andez$^*$}
\thanks{$*$\, Department of Probability and Statistics, Centro de Investigaci\'on en Matem\'aticas A.C. Calle Jalisco s/n. C.P. 36240, Guanajuato, Mexico. Email: dher@cimat.mx.  }
\author[H. A. Moreno]{Harold A. Moreno-Franco$\ddag$}
\thanks{$\ddag$\, National Research University Higher School of Economics, Shabolovkaya 31, building G, C.P. 115162,  Moscow, Russia. Email: hmoreno@hse.ru.   }
\author[J. L. P\'erez]{Jos\'e-Luis P\'erez$^\dagger$}
\thanks{$\dagger$\, Department of Probability and Statistics, Centro de Investigaci\'on en Matem\'aticas A.C. Calle Jalisco s/n. C.P. 36240, Guanajuato, Mexico. Email: jluis.garmendia@cimat.mx.  }

\begin{document}
	
\maketitle
\begin{abstract}
 In this work we study the optimal execution problem with multiplicative price impact in algorithm trading, when an agent holds an initial position of shares of a financial asset. The  inter-selling-decision times are modelled by the arrival times of a Poisson process.   The criterion to be optimised consists in  maximising the expected net present value of gains of the agent, and it is proved that an optimal  strategy  has a barrier form, depending only on the number of shares left and the level of    asset  price. 
\end{abstract}

\section{Introduction}
In this work we are interested in the problem of finding optimal execution strategies for a financial market impact model where transactions can have a permanent effect. The analysis of this problem has practical and mathematical motivations, and has been studied from different perspectives. Nowadays  the use of algorithmic trading to execute large book  orders has given rise to important questions on the best way  to execute the position, in order to decrease the negative  effect on the shift of asset  price, and also obtain the better performance of the criteria to be  optimised. In general, existence of optimal strategies can not be guaranteed, and clearly depend on the structure of the market model as well as on the parameters involved in its description. 

In any market impact model, it is crucial to describe the way that order execution algorithms will be generated. Despite the analytical tractability of the classical continuous time models, these are unfortunately not implementable in practice. On the other hand, while the models with discrete execution decision times are ideal, they lack analytical tractability, and numerical methods are required to solve them. Recently, with the aim of developing a more realistic yet analytically tractable model, random discrete execution times have been considered. Random observations are suggested in various economic literatures. See, for example, the discussion in the introduction of \cite{SS} motivated by \textit{rational inattention} (see \cite{Sims}) in macroeconomics literature. See also the discussions given in \cite{BL,LE}    for real option problems with random intervention times, and \cite{DW, GL} for  applications in optimal stopping problems and Bermudan look-back option pricing. On this regard, an important motivation for considering the Poissonian interarrival model is its potential applications in approximating the constant interarrival time cases.  It is known, in the mathematical finance literature, that randomization techniques (see, e.g.\ \cite{Carr}) are efficient in approximating constant maturity problems with those with Erlang-distributed maturities. In particular, for short maturity cases, it is known empirically that accurate approximations can be obtained by simply replacing the constant with an exponential random variable \cite{Le}.

In this note we propose a random clock, attached to the jumps of a Poisson process, for the times when the execution decisions will  take place. This is a new instrument that may represent advantages from the implementation perspective, since the randomness provided by the random clock  included in the execution strategy  introduce an additional  unpredictable   structure to the strategies. The empirical justification of this model can be approached from the following perspective, related to market micro-structural factors. It is well known that the dynamics and the volume of the trades interact with the evolution of the market liquidity of an asset. The design of each portfolio is based on the information arriving from the order flow of buys and sells decisions of other investors, but not on who is behind each decision.  This suggest that a strategic sequential trading to execute a large book order  is ``event based", represented by the new information inferred regarding the value of the asset from the composition and existence of trades from the market participants.  Thus, we are proposing that the arrival process is linked with market parameters, like liquidity, volume, market depth and order flows.  Interestingly, this view point reflects the fact, well understood in high frequency trading for instance, that time has a different meaning when we are operating an algorithmic trading strategy using cycles depending on the amount of information received, instead of the measurements based on chronological time.   

The benchmark models assume that either  trading can be done in continuous time or there are constant intervals of time at which the portfolio is balanced.   Neither of them has  practical  reasons to be sustainable, besides their analytical tractability, since investors continuously gain information about the trading environment. In the model proposed in this paper we are allowing a random clock of time in  which new information is processed,  based on the evolution of  the main factors of the market.  This is a good example of algorithmic trading which does not necessarily occur at a high frequency rate, but leaving this possibility to execute the position with frequent sells when the parameter of the Poisson process is manipulated to do so. Of course, it also helps to include asynchronous transactions to hide as much as possible the strategy followed.


 In this note  we are assuming that the agent holds a large position and, as typically happen, we expect that any selling strategy will lead to decrease of prices.   When the agent is not active the model adopted is a standard geometric Brownian motion with drift, following the work by Guo and Zervos \cite{ZG15}. Another important element in the model is related with the manner to quantify the revenues received by each selling strategy. On this regard, the criterion will be calculated as the net present value of the difference between the gains of the selling strategy and the transaction costs associated. 

In contrast with the multiplicative impact model presented here, in the seminal papers of Almgren and Chriss \cite{A-C-2}, \cite{A-C-3} and Almgren \cite{A}, the execution strategies are assumed to be absolutely continuous functions of time, having a price impact acting in an additive way; Bertsimas and Lo \cite{BL} also made fundamental contributions considering a discrete random walk model. In our case the strategies are described as Lebegue-Stieltjes integrals with respect to the paths of the Poisson process,  named {\it periodic strategies}, in analogy with the terminology used in insurance models when dividend payment decisions are taken; see,  for instance, \cite{ACW, ATW, J-X, P-X}.

More recent contributions to the theory of optimal execution  found in the literature include Huberman and Stanzl \cite{HS}, He and Mamaysky \cite{HM},
Gatheral, Schied, and Slynko \cite{GSS}, Obizhaeva and Wang \cite{OW}, Almgren and Lorenz \cite{A-L}, Engle
and Ferstenberg \cite{EF}, Schied and Sch\"oneborn \cite{SS}, Alfonsi, Fruth, and Schied \cite{AFS} \cite{AFS2}, Schied,
Sch\"oneborn, and Tehranchi \cite{SST}, Predoiu, Shaikhet, and Shreve \cite{PSS}, and Lokka \cite{L}.

In order to find an optimal strategy over the set of  periodic strategies,  
we   restrict our analysis to the set of  periodic barrier strategies. This class of barrier strategies are very easy to be implemented, since    selling decisions  are taken observing if the price of the stock lies above a certain fixed level $F$ and the remaining number of shares. Then, the  first step  consists in finding the optimal barrier strategy that maximises the performance criteria. This is done by solving the Hamilton-Jacobi-Bellman (HJB) equation associated to this problem, which allows us to obtain an explicit form of the value function for this restricted problem. Imposing a suitable smoothness   condition on the value function  we obtain the explicit value of the barrier $F_{\gamma}$ associated to the optimal strategy.  
This  strategy can be described  as follows: If the stock price is below a critical level $F_{\gamma}$ at a selling time, then it is optimal not to sell any shares. However, if the  stock price lies above the level $F_{\gamma}$ when the random clock rings,  it is optimal either to sell all available shares or liquidate a fraction of the position that will have as a consequence that  the stock price decreases.  A Verification Lemma is used to proved that the original optimisation problem within the periodic strategies can be solved implementing only barrier strategies.

The rest of the paper is organised as follows. In Section \ref{section_preliminaries}, we review the underlying model for the stock price with multiplicative price impact and provide the performance criterion, as well as the formulation of the optimal execution problem with periodic strategies.  In Section  \ref{PR}, we obtain an explicit form for the solution of the HJB equation associated to the value function over the set of periodic barrier strategies. A Verification Theorem  is provided (Theorem \ref{promax1}),  showing an explicit form for the optimal (or $\varepsilon$-optimal) periodic strategy, under appropriate conditions on the parameters of the model.   Finally, we defer the proofs of some technical lemmas to the Appendix.

\section{Market impact model}\label{section_preliminaries}
In this section we describe the optimal execution model, based on the paper by Guo and Zervos \cite{ZG15}. Let us fix a filtered probability space $(\Omega,\mathcal{F},\mathcal{F}_t,\mathcal{P})$ satisfying the usual conditions and carrying a standard $(\mathcal{F}_t)$-Brownian motion $W$ and an independent  Poisson process $N^{\gamma}$. We consider an agent holding an initial position of  $y$ shares of a financial asset, which has to be sold maximising the expected gains. The information available to the agent is enclosed in the filtration $\mathcal{F}_t$.

The trading strategies  are  denoted by the duple $(\xi_{t}^{\sell},\xi_{t}^{\barrier})$, which represents   the total number of shares that the investor has sold and bought  up  to  time $t$,  respectively. Then, the total number of shares held by the agent  at time $t$ are given by
\begin{equation}\label{sell_cond}
Y_t\eqdef y-\xi_{t}^{\sell} +\xi_{t}^{\barrier}\geq0,\qquad\text{for $t\geq0$,}
\end{equation}
where, $\xi^{\sell}$, $\xi^{\barrier}$  are  $(\mathcal{F}_t$)-adapted increasing c\`agl\`ad processes such that 
\begin{equation}\label{s.1}
\xi_{0}^{\sell}= \xi_{0}^{\barrier}=0,\quad \E\big[\expo^{4\lambda\xi_{t+}^{\barrier}}\big]<\infty\quad \text{and}\quad \displaystyle\lim_{t\to\infty}Y_{t}=0.   
\end{equation}
Although not reflected in $(\xi^{\sell},\xi^{\barrier})$, there is the restriction that   the agent cannot sell and buy shares at the same time.  The set of admissible strategies satisfying the previous conditions is denoted by $\Xi(y)$.

The stock price  observed by the agent, independently of the actions of other market participants, is modelled by the geometric Brownian motion $X^0$ with drift
\begin{equation}\label{GBM}
X^0_t=\mu X^0_t \diff t+\sigma X^0_t \diff W_t,\;\;\;X^0_0=x>0,
\end{equation}
where $\sigma\in\R$ and $\mu>0$ are constants. Let us suppose that the agent is implementing a strategy $(\xi^{\sell},\xi^{\barrier})\in \Xi(y)$. Hence, when the agent decides to sell or buy some  number of shares of the asset at time $t$, we assume that there is an impact in the price, described  as a multiplicative factor, namely, the resulting price $X_t$ is assumed to have the form
\begin{equation}\label{priceP}
X_t=X^0_t\exp\{-\lambda (\xi^{\sell}_t -\xi_{t}^{\barrier})\},
\end{equation}
for some positive constant $\lambda$ describing the permanent impact on the price, and $X^0_t$ is the solution to (\ref{GBM}). More specifically,  following  \cite{ZG15}, the controlled process dynamics  can be described  as the solution of the following stochastic differential equation
\begin{align}\label{sde1}
	dX_t=\mu X_t\diff t+\sigma X_t\diff W_t -\lambda(X_t\circ_{\sell}\diff\xi_{t}^{\sell}- X_t\circ_{\barrier}\diff\xi_{t}^{\barrier}),
\end{align}
where
\begin{equation}\label{operator}
\begin{split}
X_t\circ_{\sell}\diff\xi_{t}^{\sell}&=X_t\diff (\xi^{\sell})_t ^{\cont}+\dfrac{1}{\lambda} X_{t}(1-\expo^{-\lambda\Delta\xi_{t}^{\sell}})=X_t\diff (\xi^{\sell})_t^{\cont}+X_{t}\int_0^{\Delta\xi_{t}^{\sell}}\expo^{-\lambda u}\diff u,\\
X_t\circ_{\barrier}\diff\xi_{t}^{\barrier}&=X_t\diff (\xi^{\barrier})_t ^{\cont}+\dfrac{1}{\lambda}X_{t}(\expo^{\lambda\Delta\xi_{t}^{\barrier}}-1)=X_t\diff (\xi^{\barrier})_t^{\cont}+X_{t}\int_0^{\Delta\xi_{t}^{\barrier}}\expo^{\lambda u}\diff u,
\end{split}
\end{equation}
and the processes $(\xi^{\sell})^{\cont}$  and $(\xi^{\barrier})^{\cont}$  are the continuous part of $\xi^{\sell}$ and $\xi^{\barrier}$, respectively. The pair $(X_t, Y_t)$ is referred as the state  process associated to the strategy $(\xi^{\sell},\xi^{\barrier})$.

One of the main differences between the model introduced by Guo and Zervos \cite{ZG15} with the approach presented in this paper consists in presenting a different framework to execute the initial position $y$. While  Guo and Zervos assume that at each time $t\geq 0$ the agent should decide the number of shares to sell or buy, in this paper we assume that selling or buying can only occur at some (typically random) time points, modelled by the  jump times of an independent Poisson process $(N^{\gamma}_t\;:\;t\geq 0)$, with rate $\gamma>0$. More precisely, the selling and buying strategies   are given by
\begin{equation}\label{control}
\xi^{\sell}_t=\int_0^t\nu^{\sell}_s\diff N^{\gamma}_s\quad \text{and}\quad\xi^{\barrier}_t=\int_0^t\nu^{\barrier}_s\diff N^{\gamma}_s,
\end{equation}
where $\nu^{\sell}_t,\ \nu^{\barrier}_t$ are $\mathcal{F}_t$ adapted processes, representing  the number of shares  sold and bought at time $t$, respectively. Since the agent cannot sell and buy at the same time, the following condition holds
\begin{equation*}
\begin{cases}
\nu^{\sell}_{t}=0,&\text{if}\ \nu^{\barrier}_{t}>0,\\
\nu^{\barrier}_{t}=0,&\text{if}\ \nu^{\sell}_{t}>0,\\
\nu^{\sell}_{t}=\nu^{\barrier}_{t}=0,&\text{otherwise}.
\end{cases}
\end{equation*}	
Within this context,  selling-buying shares are necessarily done at  discrete  periods of time (there cannot be continuous selling and buying)  since selling-buying decisions can only occur when the process $N^{\gamma}$ has jumps. The set of selling-buying  decision times is denoted as $\mathcal{T}=\{T_1, T_2,\dots\}$, and the quantities $T_k-T_{k-1}$, $k\geq 0$, are the inter-selling-buying-decision times, which are exponentially distributed with mean $1/\gamma$. The number of shares sold and bought at each decision time $T_j$ are denoted by $\nu^{\sell}_{T_j}$ and $\nu^{\barrier}_{T_j}$ ,  respectively, with $\Theta^{\sell,\barrier}=\{(\nu^{\sell}_{T_1},\nu^{\sell}_{T_2},\dots) ,(\nu^{\barrier}_{T_1},\nu^{\barrier}_{T_2},\dots) \}$ representing a selling-buying  strategy via (\ref{control}); the subset of strategies $(\xi^{\sell}, \xi^{\barrier})\in \Xi(y)$ which can be represented as in  (\ref{control}) is denoted by $\mathcal{A}(y)$. For those strategies $(\xi^{\sell}, \xi^{\barrier}) \in \mathcal{A}(y)$, the operator defined in (\ref{operator}) can be written as
\begin{equation}\label{pc.1}
\begin{split}
X_t\circ_{\sell}\diff\xi_{t}^{\sell}&=\dfrac{1}{\lambda}X_{t}(1-\expo^{-\lambda\nu^{\sell}_{t}})\diff N_{t}^{\gamma}=X_{t}\bigg(\int_0^{\nu^{\sell}_{t}}\expo^{-\lambda u}\diff u\bigg) \diff N_{t}^{\gamma},\\
X_t\circ_{\barrier}\diff\xi_{t}^{\barrier}&=\dfrac{1}{\lambda}X_{t}(\expo^{\lambda\nu_{t}^{\barrier}}-1)\diff N_{t}^{\gamma}=X_{t}\bigg(\int_0^{\nu_{t}^{\barrier}}\expo^{\lambda u}\diff u\bigg)\diff N_{t}^{\gamma}.
\end{split}
\end{equation}

Let $C_{\sell}, C_{\barrier}$ be positive constants representing the transaction cost associated with the sell and buy of shares, respectively. Then, the gains associated with each strategy $(\xi^{\sell}, \xi^{\barrier})\in \mathcal{A}(y)$ is given by
$$
\int_0^\infty  \left(X_t\circ_{\sell}\diff\xi^{\sell}_t -X_t\circ_{\barrier}\diff\xi^{\barrier}_t-C_{\sell}\diff\xi_{t}^{\sell} -C_{\barrier}\diff\xi_{t}^{\barrier}\right),
$$
and the agent's  objective is to maximise the expected  net present value of gains
\begin{equation}\label{pc}
J_{x,y}(\xi^{\sell}, \xi^{\barrier})\eqdef\E_{x}\biggr[\int_{0}^{\infty}\expo^{-\delta t}\left(X_t\circ_{\sell}\diff\xi^{\sell}_t-X_t\circ_{\barrier}\diff\xi^{\barrier}_t-C_{\sell}\diff\xi_{t}^{\sell}-C_{\barrier}\diff\xi_{t}^{\barrier}\right)\biggl],
\end{equation}
over the set $ \mathcal{A}(y)$.
The parameter $\delta>0$ is the discount factor, and we assume  that $\delta>\mu$ in order to avoid arbitrage opportunities, as described in \cite[Proposition 3.4]{ZG15}. Given an initial condition $(x,y)\in\R_+\times\R_+$, we say that $(\xi^{\sell*},\xi^{\barrier*})\in \mathcal{A}_(y)$ is an optimal strategy if and only if 
\begin{equation*}
J_{x,y}(\xi^{\sell}, \xi^{\barrier})\leq J_{x,y}(\xi^{\sell*}, \xi^{\barrier*}),\ \text{for all}\ (\xi^{\sell}, \xi^{\barrier})\in \mathcal{A}(y).
\end{equation*}
The value function of this stochastic control problem is defined as
\begin{equation}\label{ValueFunction}
u_0(x,y)=\sup_{(\xi^{\sell}, \xi^{\barrier})\in\mathcal{A}(y)}J_{x,y}(\xi^{\sell},\xi^{\barrier}).
\end{equation}


\subsection{Regularity property of the model}
A remarkable property of our model is the {\it regularity} \cite{GS}, which is understood as the requirement that, first,  the optimisation problem  (\ref{ValueFunction}) has an optimal solution and,
second, as the initial position $y$ is positive, it should be  expected that the optimal execution strategy does not involve buying decisions along the time needed to liquidate the initial position.  In the rest of this section we elaborate along the second condition, while the first one will be  approached in the next section. 

Consider strategies where the agent only sell shares, i.e., strategies of the form $(\xi^{\sell},\xi^{\barrier})\in\mathcal{A}(y)$ such that $\xi^{\barrier}\equiv0$.   The subset of strategies $(\xi^{\sell},0)\in\mathcal{A}(y)$ is denoted by $\mathcal{A}^{\sell}(y)$, whose elements are represented only by $\xi^{\sell}$, and the value function $u(x,y)$ for this problem is given by
\begin{equation}\label{vf.1}
u(x,y)\eqdef\sup_{(\xi^{\sell},0)\in\mathcal{A}(y)}J_{x,y}(\xi^{\sell},0)=\sup_{\xi^{\sell}\in\mathcal{A}^{\sell}(y)}\E_{x}\biggr[\int_{0}^{\infty}\expo^{-\delta t}\left(X_t\circ_{\sell}\diff\xi^{\sell}-C_{\sell}\diff\xi_{t}^{\sell}\right)\biggl].
\end{equation}

It is clear that  $u(x,y)\leq u_0(x,y)$, since $\mathcal{A}^{\sell}(y)\subset\mathcal{A}(y)$; the equality between these value functions is established in the next result.

\begin{proposition}\label{opt_sell}
Let $u_0$ and $u$ be the value functions given in \eqref{ValueFunction} and \eqref{vf.1}, respectively. Then, $u=u_{0}$.
\end{proposition}
In order to prove this  result we need to use a technical tool described in the next   lemma, which proof is presented in the Appendix.
\begin{lemma}\label{sell.1}
	For each $(\xi^{\sell},\xi^{\barrier})\in\mathcal{A}(y)$, there exists $\bar{\xi}^{\sell}\in\mathcal{A}^{\sell}(y)$ such that $\xi^{\sell}_{t}-\xi^{\barrier}_{t}\leq\bar{\xi}^{\sell}_{t}\leq\xi^{\sell}_{t}$ for all $t\geq0$. 
\end{lemma}
 \begin{proof}[Proof of Proposition \ref{opt_sell}]
Recall that we only need to show  the inequality $u_{0}(x,y)\leq u(x,y)$. Let $(\xi^{\sell},\xi^{\barrier})$ be in $\mathcal{A}(y)$ and consider the strategy $\bar{\xi}^{\sell}\in\mathcal{A}^{\sell}(y)$ as in Lemma \ref{sell.1}.  Now, we consider the processes $X_{t}=X^{0}_{t}\exp\{-\lambda(\xi^{\sell}_{t}-\xi^{\barrier}_{t})\}$ and $\overline{X}_{t}=X^{0}_{t}\exp\{-\lambda\bar{\xi}^{\sell}_{t}\},$ which satisfy the following SDEs,  
\begin{equation}\label{s.2.1}
\begin{split}
dX_t&=\mu X_t\diff t+\sigma X_t\diff W_t -\lambda(X_t\circ_{\sell}\diff\xi_{t}^{\sell}- X_t\circ_{\barrier}\diff\xi_{t}^{\barrier}),\\
d\overline{X}_t&=\mu \overline{X}_t\diff t+\sigma \overline{X}_t\diff W_t -\lambda\overline{X}_t\circ_{\sell}\diff\bar{\xi}_{t}^{\sell},
\end{split}
\end{equation}
respectively, where $X_t\circ_{\sell}\diff\xi_{t}^{\sell}$, $X_t\circ_{\barrier}\diff\xi_{t}^{\barrier}$ and $\overline{X}_t\circ_{\sell}\diff\bar{\xi}_{t}^{\sell}$ are given as in \eqref{pc.1}. Applying integration by parts in $\expo^{-\delta(t\wedge\bar{\tau}_{m})}X_{t\wedge\bar{\tau}_{m}}$, where $\bar{\tau}_{m}\eqdef\inf\{t>0:X^{0}_{t}>m\}$, and  using \eqref{s.2.1}, it follows
\begin{align}\label{s.3}
\int_{0}^{t\wedge\bar{\tau}_{m}}&\expo^{-\delta t}\left(X_t\circ_{\sell}\diff\xi^{\sell}_t-X_t\circ_{\barrier}\diff\xi^{\barrier}_t-C_{\sell}\diff\xi_{t}^{\sell}-C_{\barrier}\diff\xi_{t}^{\barrier}\right)\\
=&\dfrac{1}{\lambda}\big(x-\expo^{-\delta(t\wedge\bar{\tau}_{m})}X^{0}_{t\wedge\bar{\tau}_{m}}\expo^{-\lambda(\xi^{\sell}_{t\wedge\bar{\tau}_{m}}-\xi^{\barrier}_{t\wedge\bar{\tau}_{m}})}\big)-\frac{\delta-\mu}{\lambda}\int_{0}^{t\wedge\bar{\tau}_{m}}\expo^{-\delta s}X_{s}^{0}\expo^{-\lambda(\xi^{\sell}_{s}-\xi^{\barrier}_{s})}\diff s\notag\\
&\quad-\int_{0}^{t\wedge\bar{\tau}_{m}}\expo^{-\delta s}(C_{\sell}\diff\xi_{s}^{\sell}+C_{\barrier}\diff\xi^{\barrier}_{s})+\frac{\sigma}{\lambda}\int_{0}^{t\wedge\bar{\tau}_{m}}\expo^{-\delta s}X_{s}^{0}\expo^{-\lambda(\xi^{\sell}_{s}-\xi^{\barrier}_{s})}\diff W_{s}\notag\\
\leq &\dfrac{1}{\lambda}\big(x-\expo^{-\delta(t\wedge\bar{\tau}_{m})}X^{0}_{t\wedge\bar{\tau}_{m}}\expo^{-\lambda\bar{\xi}^{\sell}_{t\wedge\bar{\tau}_{m}}}\big)-\frac{\delta-\mu}{\lambda}\int_{0}^{t\wedge\bar{\tau}_{m}}\expo^{-\delta s}X_{s}^{0}\expo^{-\lambda\bar{\xi}^{\sell}_{s}}\diff s\notag\\
&-\int_{0}^{t\wedge\bar{\tau}_{m}}\expo^{-\delta s}C_{\sell}\diff\bar{\xi}_{s}^{\sell}+\frac{\sigma}{\lambda}\int_{0}^{t\wedge\bar{\tau}_{m}}\expo^{-\delta s}X_{s}^{0}\expo^{-\lambda(\xi^{\sell}_{s}-\xi^{\barrier}_{s})}\diff W_{s}\notag\\
=&\int_{0}^{t\wedge\bar{\tau}_{m}}\expo^{-\delta s}(\overline{X}_{s}\circ_{s}\diff\bar{\xi}^{\sell}_{s}-C_{\sell}\diff\bar{\xi}_{s}^{\sell})+\frac{\sigma}{\lambda}\int_{0}^{t\wedge\bar{\tau}_{m}}\expo^{-\delta s}X_{s}^{0}(\expo^{-\lambda(\xi^{\sell}_{s}-\xi^{\barrier}_{s})}-\expo^{-\lambda\bar{\xi}_{s}^{\sell}})\diff W_{s}\notag,
\end{align}
Note that $\displaystyle\int_{0}^{t\wedge\bar{\tau}_{m}}\expo^{-\delta s}X_{s}^{0}(\expo^{-\lambda(\xi^{\sell}_{s}-\xi^{\barrier}_{s})}-\expo^{-\lambda\bar{\xi}_{s}^{\sell}})\diff B_{s}$ is a square-martingale, since by hypothesis \\ $\E[\expo^{4\lambda\xi^{\barrier}_{(t\wedge\bar{\tau}_{m})+}}]<\infty$, whose expected value is zero; see more in \cite[p. 293]{ZG15}. Then, taking expected value in \eqref{s.3}, it yields 
\begin{multline*}
\E_{x}\bigg[\int_{0}^{t\wedge\bar{\tau}_{m}}\expo^{-\delta t}\left(X_t\circ_{\sell}\diff\xi^{\sell}_t -X_t\circ_{\barrier}\diff\xi^{\barrier}_t-C_{\sell}\diff\xi_{t}^{\sell} -C_{\barrier}\diff\xi_{t}^{\barrier}\right)\bigg]\\
\leq\E_{x}\bigg[\int_{0}^{t\wedge\bar{\tau}_{m}}\expo^{-\delta s}(\overline{X}_{s}\circ_{s}\diff\bar{\xi}^{\sell}_{s}-C_{\sell}\diff\bar{\xi}_{s}^{\sell})\bigg].
\end{multline*}
Then, letting $t,m\rightarrow\infty$ and using the Dominated Convergence Theorem, we get
\begin{equation}\label{s.5}
J_{x,y}(\xi^{\sell},\xi^{\barrier})\leq J_{x,y}(\bar{\xi}^{\sell},0)\leq u(x,y). 
\end{equation}
Since the above inequality holds for each $(\xi^{\sell},\xi^{\barrier})\in\mathcal{A}(y)$, we conclude that $u_{0}(x,y)\leq u(x,y)$, completing the proof of this result.  \end{proof}

\section{HJB equations and optimal execution}\label{PR}
The remainder of this paper is devoted to the presentation of an optimal solution to the  execution problem (\ref{vf.1}),  paying particular attention to the structure of this strategy.  Roughly, we look for  optimal selling strategies with a  simple structure, facilitating its adaptability from the practical view point. One of the main results of this work establishes the existence of an optimal strategy $\hat{\xi}^{\sell}$, which has a barrier form in the state space.
A {\it barrier strategy} is described in terms of the remaining number of shares to be sold and the level of the observed price at each period of time. These are compared with a mark, which is decided from the beginning and depends on a non-negative constant $F$, referred hereafter as a {\it periodic barrier}.

More precisely, fixing a periodic barrier $F >0$, the number of shares to be  sold in the $i$-th  arrival time $T_i$ of  the Poisson process  $N^{\gamma}$,  is given by
\begin{equation}\label{Bstrat1}
\nu_{F}(T_i)\eqdef Y_{T_i}\wedge \dfrac{1}{\lambda}(\ln X_{T_i}-\ln F)^+,
\end{equation}
where  $(X_{T_i}, Y_{T_i})$ is the position of the state process at the arrival time $T_i$ and $(\ln X_{T_i}-\ln F)^+\eqdef\max\{0,(\ln X_{T_i}-\ln F)\}$. This type of strategies is denoted by $\xi^{\sell, F}$, and the set of these strategies is defined by $\mathcal{A}^{\sell}_{\barrier}(y)$, which is clearly   a subset of    $\mathcal{A}^{\sell}(y)$. The corresponding value function for this set of strategies  is defined as
\begin{equation}\label{ValueFunction2}
u_{\barrier}(x,y)=\sup_{\xi^{\sell, F}\in \mathcal{A}^{\sell}_{\barrier}(y)} J_{x,y}(\xi^{\sell, F},0),
\end{equation}
with $J_{x,y}$  as in \eqref{pc}.

In order to relate the above value function $u_{\barrier}$ with the original one $u$, defined in (\ref{vf.1}), it is convenient  to provide a brief description of the approach to be followed. As a first step, we shall solve the problem (\ref{ValueFunction2}) using dynamic programming techniques.  Noting that the periodic strategies are described by a single parameter, the crucial point consists in proving that there is a parameter $F_\gamma$, defined below by (\ref{F1}), for which there exists a smooth solution $v$ to the HJB equation (\ref{HJBin2}) associated with $u_{\barrier}$; see Proposition \ref{Tcons1}.  As a second step,  using the parameter $F_\gamma$ a periodic strategy  $\xi^{\sell, F_\gamma}$ can be  built using the expression (\ref{Bstrat1}).  Finally, using the HJB equation associated with the original value function $u$ described in (\ref{HJBP}), in the last step it is proved that $v$ solves that equation, and that $\xi^{\sell, F_\gamma}$ is optimal  for $u$ within the set $\mathcal{A}^{\sell}(y)$, concluding that $u_{\barrier}= u$; see the Verification Theorem \ref{LemVer1}.


\begin{remark}
By standard dynamic programming arguments, it is well known that the value functions $u$  and $u_{\barrier}$, are associated to HJB equations, which are given, respectively, by
\begin{equation}\label{HJBP}
\begin{cases}
\mathcal{L}w(x,y)+\displaystyle\max_{0\leq l\leq y}\{\gamma G(x,y,l;w)\}=0,& \text{for all}\ x>0\ \text{and}\ y>0, \\
w(x,y)=0,&\text{for all}\ x>0\ \text{and}\ y=0,
\end{cases}
\end{equation}
\begin{equation}\label{HJBin2}
\begin{cases}
 \mathcal{L}v(x,y)+\gamma G\biggr(x,y,\biggl[y\wedge \dfrac{1}{\lambda}\ln(x/F_{\gamma})^+\biggr];v\biggl)= 0,&\text{for all}\ x>0\ \text{and}\ y>0,\\
 v(x,y)=0,&  \text{for all}\ x>0\ \text{and}\ y=0,
 \end{cases}
\end{equation}
where $F_{\gamma}$ is a positive constant which will be determined later on. Here, the operators $\mathcal{L}$ and $G$ are defined by
\begin{align}
\mathcal{L}f(x,y)&\eqdef\dfrac{1}{2}\sigma^2x^2f_{xx}(x,y)+\mu xf_x(x,y)-\delta f(x,y),\label{Lap}\\
G(x,y,l;f)&\eqdef f(x\expo^{-\lambda l},y-l)-f(x,y)+\dfrac{1}{\lambda}(1-\expo^{-\lambda l})x-C_{\sell}l.\label{maxop1}
\end{align}
\end{remark}

\subsection{Construction and regularisation of the solution $v$ to the HJB equation (\ref{HJBin2})}\label{Z3}

Observe that   we can simplify the HJB equation \eqref{HJBin2} according with the following three different scenarios:
 \begin{itemize}
 	\item [(i)] When $x<F_{\gamma}$, this restriction  corresponds to the {\it waiting region} $\mathcal{W}$ because the price is too low for selling any shares to be optimal, and therefore \eqref{HJBin2} takes the form
 	\begin{equation}\label{HJB0}
 	\mathcal{L}v(x,y)=0.
 	\end{equation}
 	 	\item[(ii)] When $F_{\gamma}\leq x< F_{\gamma}\expo^{\lambda y}$, the agent takes and intermediate position of selling $\mathbb{Y}(x)\eqdef\dfrac{1}{\lambda}\ln(x/F_{\gamma})$  assets. Now since $\expo^{-\lambda\mathbb{Y}(x)}=\dfrac{F_{\gamma}}{x}$,  \eqref{HJBin2} can be written as
 	 \begin{align}\label{HJBin5}
 	 	\mathcal{L}_{\gamma}v(x,y)&+\gamma\biggr[v\left(F_{\gamma},
 	 	y-\mathbb{Y}(x)\right)+\dfrac{x-F_{\gamma}}{\lambda}-C_{\sell}\mathbb{Y}(x)\biggr]= 0,
 	 \end{align}
where
 	\begin{equation*}
 	\mathcal{L}_{\gamma}v(x,y)\eqdef\dfrac{1}{2}\sigma^2x^2v_{xx}(x,y)+\mu xv_x(x,y)-(\delta+\gamma) v(x,y).
 	\end{equation*}
 	\item[(iii)]  Finally, when $x\geq F_{\gamma}\expo^{\lambda y}$, we have that the asset price is sufficiently high and the corresponding decision is to execute the complete set of assets available, and then \eqref{HJBin2} is reduced to
 	\begin{align}\label{HJBin4}
 		 \mathcal{L}_{\gamma}v(x,y)+\gamma\biggl[
 		 \dfrac{1}{\lambda}(1-\expo^{-\lambda y})x-C_{\sell}y\biggr]=0.
 	\end{align}
\end{itemize}
 We  obtain explicit solutions for each  one of the three regions, which are described  in \eqref{HJBg1}. A quite important issue in the form of the solutions proposed below is the smoothness at the boundary of each region, which derives in obtaining  an explicit form of the general solution. The proof of the following result is given in the Appendix.
\begin{proposition}\label{Tcons1}
The HJB equation \eqref{HJBin2} has a solution $v$, which belongs to $\hol^{2,1}(\R_{+}\times\R_{+})$ and	 is given by
\begin{equation}\label{HJBg1}
v(x,y)=
\begin{cases}
\displaystyle 0, &\text{if}\ y=0\ \text{and}\ x>0,\\
\displaystyle\dfrac{(F_{\gamma}-C_{\sell})(1-\expo^{-\lambda ny})x^{n}}{\lambda n F_{\gamma}^{n}},& \text{if}\ y>0\ \text{and}\ x<F_{\gamma},\\
\displaystyle A_{\gamma}\biggl(\dfrac{x}{F_{\gamma}}\biggr)^{m_{\gamma}}-\dfrac{(F_{\gamma}-C_{\sell})\expo^{-\lambda ny}x^{n}}{\lambda nF_{\gamma}^{n}}&\\
\hspace{3cm}+\dfrac{\gamma x}{\lambda(\eta+\gamma)}-\dfrac{\gamma C_{\sell}\ln x}{\lambda(\delta+\gamma)}+C_{\gamma}, &\text{if}\ y>0\ \text{and}\ F_{\gamma}\leq x< F_{\gamma}\expo^{\lambda y},\\
\displaystyle \dfrac{A_{\gamma}(1-\expo^{-\lambda m_{\gamma}y})x^{m_{\gamma}}}{F_{\gamma}^{m_{\gamma}}}+\dfrac{\gamma x(1-\expo^{-\lambda y})}{\lambda(\eta+\gamma)}-\dfrac{\gamma C_{\sell}y}{\delta+\gamma},&\text{if}\ y>0\ \text{and}\  x\geq F_{\gamma}\expo^{\lambda y} ,
\end{cases}
\end{equation}
where $\eta\eqdef\delta-\mu$,
\begin{align}
F_{\gamma}&\eqdef \dfrac{\displaystyle\dfrac{C_{\sell}}{\delta+\gamma}\biggl(\delta-m_{\gamma}\biggl(\dfrac{\delta}{n}+\dfrac{\gamma b}{\delta+\gamma}\biggl)\biggl)}{\displaystyle\dfrac{\eta}{\eta+\gamma}-\dfrac{m_{\gamma}}{\delta+\gamma}\biggl(\dfrac{\delta}{n}-\dfrac{\gamma\mu}{\eta+\gamma}\biggl)},\label{F1}\\
A_{\gamma}&\eqdef\dfrac{F_{\gamma}}{\lambda (\delta+\gamma)}\biggr(\dfrac{\delta}{n}-\dfrac{\gamma\mu}{\eta+\gamma}\biggl)-\dfrac{C_{\sell}}{\lambda(\delta+\gamma)}\biggl(\dfrac{\delta}{n}+\dfrac{\gamma b}{\delta+\gamma}\biggr),\label{Cg1}\\
C_{\gamma}&\eqdef \dfrac{\gamma(F_{\gamma}-C_{\sell})}{\lambda n(\delta+\gamma)}+ \dfrac{\gamma}{\lambda(\delta+\gamma)}\biggl(\dfrac{bC_{\sell}}{\delta+\gamma}+C_{\sell}\ln F_{\gamma}-F_{\gamma}\biggr),\label{Cg2}
\end{align}
and $b\eqdef\dfrac{1}{2}\sigma^{2}-\mu$. The constants $n$, $m_{\gamma}$  are the positive and negative solutions  to
\begin{align}
 &\dfrac{1}{2}\sigma^2l^{2}-b l-\delta=0,\label{pos1}\\
&\dfrac{1}{2}\sigma^2l^{2}-b l-(\delta+\gamma)=0,\label{neg1}
\end{align}
respectively.
\end{proposition}

\begin{remark}
The fact that $\delta>\mu$ implies that the solution $n$ to \eqref{pos1} satisfies that $n>1$, for all $\sigma\in\R$.  
\end{remark}

\subsection{Equivalence between the HJB equations}\label{OE1}
The rest of this paper  is dedicated to verify that the strategy given in \eqref{Bstrat1}, with barrier $F_{\gamma}$, defined above   in \eqref{F1}, is optimal within the set of strategies $\mathcal{A}^{\sell}(y)$, and that the function $v$ given in \eqref{HJBg1}, satisfies the HJB equation \eqref{HJBP}. To this end we need  the following technical results,  whose proof is given in the Appendix in order to introduce in a succinct way their  main consequences.

 \begin{lemma}\label{ac1}
	Let $a_{\gamma}$ be defined by
	\begin{equation}\label{a1}
a_{\gamma}\eqdef\dfrac{\displaystyle\dfrac{1}{\delta+\gamma}\biggl(\delta-m_{\gamma}\biggl(\dfrac{\delta}{n}+\dfrac{\gamma b}{\delta+\gamma}\biggl)\biggl)}{\displaystyle\dfrac{\eta}{\eta+\gamma}-\dfrac{m_{\gamma}}{\delta+\gamma}\biggl(\dfrac{\delta}{n}-\dfrac{\gamma\mu}{\eta+\gamma}\biggl)},
\end{equation}
where $\gamma>0$. Then, for each $\gamma>0$, $ 1< a_{\gamma}< \dfrac{n}{n-1}$, and it satisfies the following asymptotic limits
\begin{equation*}
\begin{cases}
a_{\gamma}\rightarrow1,& \text{when}\ \gamma\rightarrow0,\\
a_{\gamma}\rightarrow\dfrac{n}{n-1},& \text{when}\ \gamma\rightarrow\infty.
\end{cases}
\end{equation*}
\end{lemma}

\begin{proposition}\label{promax1}
Let $v$ be as in \eqref{HJBg1}. Then, for each $(x,y)\in\mathbb{R}_+\times\mathbb{R}_+$  the following identity holds
\begin{align}\label{c1}
\max_{0\leq l\leq y}G(x,y,l;v)&=G\biggr(x,y,\biggl[y\wedge \dfrac{1}{\lambda}\ln(x/F_{\gamma})^+\biggr];v\biggl),
\end{align}
with $G$ defined as in \eqref{maxop1}.
\end{proposition}

\begin{remark}
 Note that putting together Propositions \ref{Tcons1} and \ref{promax1}, it follows immediately that $v$ is a solution to the HJB equation \eqref{HJBP}.
\end{remark}

Next result  identifies the solution of the HJB equation \eqref{HJBg1} with the value function $u$, providing also an optimal strategy within the set $\mathcal{A}^{\sell}(y)$.

\begin{theorem}[Verification Theorem]\label{LemVer1}
Consider the periodic optimal execution problem formulated in Section 2 and the function $v$ defined by \eqref{HJBg1}. Then, $v$ agrees with the value function $u$ of the periodic stochastic control problem. In particular,
\begin{align*}
u(x,y)=\sup_{\xi^{\sell}\in\mathcal{A}^{\sell}(y)}J_{x,y}(\xi^{\sell},0)=v(x,y)\qquad\text{for all $(x,y)\in\mathbb{R}_+\times\mathbb{R}_+$.}
\end{align*}
Furthermore,  define the strategy $\xi^{\sell*}_{t}=\displaystyle\int_{0}^{t}\nu_{s}^*d N^{\gamma}_{s}$ with
\begin{equation}\label{optesbar1}
\nu_{T_i}^*=Y_{T_i}^*\wedge\frac{1}{\lambda}\left(\ln X_{T_i}^*-\ln F_{\gamma}\right)^+,
\end{equation}
where $F_{\gamma}$ is as in \eqref{F1}, $\mathcal{T}=\{T_i\}_{i=1}^{\infty}$ is the set of arrival times of the Poisson process $N^{\gamma}$, and $(X^* ,Y^*)$ is the state process   associated with the liquidation strategy $\xi^{\sell*}$. Then, the following statements hold
\begin{itemize}
\item [(i)] If $\mu-\frac{1}{2}\sigma^2\geq 0$, then $\xi^{\sell*}$ is an optimal periodic liquidation strategy.
\item [(ii)] If $\mu-\frac{1}{2}\sigma^2< 0$, then $\xi^{\sell*}$ is not an optimal periodic liquidation strategy. So if we define 
\begin{equation}
\label{epsilon-optimal}
\xi^{\sell*j}_t=\xi^{\sell*}_t1_{\{t\leq j\}}+y 1_{\{j< t\}},\qquad \text{ for $t>0$ and $j\geq 1$,}
\end{equation}
then $\{\xi^{\sell*j}\}_{j=1}^{\infty}$ is a sequence of $\varepsilon$-optimal periodic strategies.
\end{itemize}

\end{theorem}
 \begin{proof}\label{ver_lemma}
 Take  $(x,y)\in\R_{+}^{*}\times\R_+$ an initial condition, $\xi^{\sell}\in \mathcal{A}^{\sell}(y)$, described by the following selling strategy $\Theta=\{\nu_{T_1},\nu_{T_2},\dots\}$, where $\mathcal{T}=\{T_1,T_2,\dots\}$ is the set of selling times, and  $(\tau_{m})_{m\in\mathbb{N}}$ the sequence of stopping times defined by $\tau_m :=\inf\{t>0:X_t > m\}$. Using It\^o-Tanaka-Meyer's  formula and the left continuity of the processes $X$ and $Y$, we can see 
\begin{align*}
 \expo^{-\delta (t\wedge \tau_{m})}v(X_{t\wedge \tau_{m}},Y_{t\wedge \tau_{m}})&=v(x,y)+\int_0^{t\wedge \tau_{m}}\expo^{-\delta s}\mathcal{L}v(X_s,Y_s)\diff s+M_{t\wedge \tau_{m}}\\
 &\quad+\sum_{0\leq s\leq t\wedge \tau_{m}}\expo^{-\delta s}[v(X_{s+},Y_{s+})-v(X_s,Y_s)],
 \end{align*}
 where  $M_{t\wedge \tau_{m}}\eqdef\sigma\displaystyle\int_0^{t\wedge \tau_{m}}\expo^{-\delta s}X_sw_x(X_s,Y_s)\diff W_s$. On the other hand
 \begin{align*}
 &\sum_{0\leq s\leq t\wedge \tau_{m}}\expo^{-\delta s}[v(X_{s+},Y_{s+})-v(X_s,Y_s)]\notag\\
 &\quad=\int_0^{t\wedge \tau_{m}}\expo^{-\delta s}\left[v(X_{s-}\expo^{-\lambda \nu_s},Y_{s-}-\nu_s)-v(X_{s-},Y_{s-})\right]\diff N^{\gamma}_{s}\notag\\
 &\quad=\int_0^{t\wedge \tau_{m}}\expo^{-\delta s} G(X_{s-},Y_{s-},\nu_{s};v)\diff N^{\gamma}_{s}-\int_0^{t\wedge \tau_{m}}\expo^{-\delta s}\left[\dfrac{1}{\lambda}(1-\expo^{-\lambda \nu_s})X_{s-}-C_{\sell}\nu_s\right]\diff N^{\gamma}_{s}\notag\\
 &\quad=H_{t\wedge \tau_{m}}+\int_0^{t\wedge \tau_{m}}\gamma \expo^{-\delta s} G(X_{s-},Y_{s-},\nu_{s};v)\diff s-\int_0^{t\wedge \tau_{m}} \expo^{-\delta s}\left[\dfrac{1}{\lambda}(1-\expo^{-\lambda \nu_s})X_{s-}-C_{\sell}\nu_s\right]\diff N^{\gamma}_{s},
 \end{align*}
 with $G$ as in \eqref{c1}, $H_{t\wedge \tau_{m}}:=\displaystyle\int_0^{t\wedge \tau_{m}}\expo^{-\delta s} G(X_{s-},Y_{s-},\nu_{s};v)\diff\widetilde{N}^{\gamma}_{s}$ and $\widetilde{N}$ is the compensated Poisson process. Hence, putting these pieces together and observing that
\begin{equation*}
 \mathcal{L}v(X_s,Y_s)+\gamma G(X_{s-},Y_{s-},\nu_{s};v)\leq \mathcal{L}v(X_s,Y_s)+\max_{0\leq l\leq y}\{\gamma G(X_{s-},Y_{s-},l;v)\}=0,
\end{equation*}
we have
 \begin{align}\label{aux_1}
 \expo^{-\delta (t\wedge \tau_{m})}v(X_{t\wedge \tau_{m}},Y_{t\wedge \tau_{m}})&=v(x,y)+M_{t\wedge \tau_{m}}+H_{t\wedge \tau_{m}}\\
 &\quad-\int_0^{t\wedge \tau_{m}}\expo^{-\delta s}\left[\dfrac{1}{\lambda}(1-\expo^{-\lambda \nu_s})X_{s-}-C_{\sell}\nu_s\right]\diff N^{\gamma}_{s}\notag\\
 &\quad+\int_0^{t\wedge \tau_{m}} \expo^{-\delta s}[\mathcal{L}v(X_s,Y_s)+\gamma G(X_{s-},Y_{s-},\nu_{s};v)]\diff s\notag\\
 &\leq v(x,y)+M_{t\wedge \tau_{m}}+H_{t\wedge \tau_{m}}\notag\\
 &\quad-\int_0^{t\wedge \tau_{m}}\expo^{-\delta s}\biggr[\dfrac{1}{\lambda}(1-\expo^{-\lambda \nu_s})X_{s-}-C_{\sell}\nu_s\biggl]\diff N^{\gamma}_{s}.\notag
 \end{align}
 From \eqref{HJBg1} it is not difficult to see that there exists a positive constant $K$ such that,
 \begin{equation}\label{growth}
 |v(x,y)|\leq K(1+x+y), \text{for all $(x,y)\in\mathbb{R}_+\times\mathbb{R}_+$,}
 \end{equation}
 hence it follows that the processes $(M_{t \wedge \tau_{m}}; t\geq0 )$ and $(H_{t \wedge \tau_{m}}; t\geq0 )$ are zero-mean $\mathbb{P}$-martingales. Then, taking expectations in \eqref{aux_1} we obtain
 \begin{align}\label{HJbdes2}
 v(x,y)&\geq \E_{x}\left[\expo^{-\delta (t\wedge \tau_{m})}v(X_{t\wedge \tau_{m}},Y_{t\wedge \tau_{m}})\right]\\
 &\quad+\E_{x}\left[\int_0^{t\wedge \tau_{m}}\expo^{-\delta s}\left[\dfrac{1}{\lambda}(1-\expo^{-\lambda \nu_s})X_{s-}-C_{\sell}\nu_s\right]\diff N^{\gamma}_{s}\right].\notag
 \end{align}
 Now, from  the expression of the process $X_t$ in (\ref{priceP}) and recalling that $\delta>\mu$, we have that
 \begin{align}\label{HJbdes3}
 \lim_{t,m\longrightarrow\infty} \E_{x}[\expo^{-\delta (t\wedge \tau_{m})}v(X_{t\wedge \tau_{m}},Y_{t\wedge \tau_{m}})]&\leq \lim_{t,m\longrightarrow\infty} \E_{x}[\expo^{-\delta (t\wedge \tau_{m})}K(1+X_{t\wedge \tau_{m}}+Y_{t\wedge \tau_{m}})]\\
 &\leq \lim_{t,m\longrightarrow\infty} \E_{x}[\expo^{-\delta (t\wedge \tau_{m})}K(1+X_{t\wedge \tau_{m}}+y)]=0.\notag
 \end{align}
Using \eqref{sell_cond} we note the following
 \begin{equation}\label{int_cond}
 \int_0^{\infty}\nu_s\diff N^{\gamma}_{s}=\lim_{t\to\infty} \xi^{\sell}_t\leq y.
 \end{equation}
 Then, letting $m,t\rightarrow\infty$ in \eqref{HJbdes2}, using \eqref{HJbdes3}, \eqref{int_cond} and Monotone Convergence Theorem, we obtain that
 \begin{align*}
 v(x,y)&\geq \E_{x}\left[\int_0^{\infty}\expo^{-\delta s}\left[\dfrac{1}{\lambda}(1-\expo^{-\lambda \nu_s})X_{s-}-C_{\sell}\nu_s\right]\diff N^{\gamma}_{s}\right]\\
 &=\E_{x}\left[\int_{0}^{\infty}\expo^{-\delta  s}\left[X_s\circ_{\sell}\diff \xi^{\sell}_s-C_{\sell}\diff\xi_{s}^{\sell}\right]\right].
 \end{align*}
 Taking the maximum over all $\xi^{\sell}\in\mathcal{A}^{\sell}(y)$ we conclude that $u(x,y)\leq v(x,y)$. Let $(X^* ,Y^*)$ be the state process   associated with the liquidation strategy $\xi^{\sell*}$, given by $\xi^{\sell*}_t=\displaystyle\int_0^t\nu^*_{s}\diff N^{\gamma}_{s}$, with $\nu^*$ as in \eqref{optesbar1}. Note that $\xi^{\sell}$ is admissible as long as $\displaystyle\lim_{t\to\infty}\xi^{\sell}_t=y$. We can easily check using \eqref{priceP} that this indeed the case if and only if $\mu-\frac{1}{2}\sigma^2\geq 0$ since $\displaystyle\limsup_{ t\rightarrow\infty} X_t^0=\infty$. Proceeding in a similar way as in \eqref{aux_1}, we get that
 \begin{align}\label{aux_4}
 \expo^{-\delta (t\wedge \bar{\tau}_{m})}v(X_{t\wedge \bar{\tau}_{m}}^*,Y_{t\wedge \bar{\tau}_{m}}^*)&=v(x,y)+M_{t\wedge \bar{\tau}_{m}}^*+H_{t\wedge \bar{\tau}_{m}}^*\\
 &\quad-\int_0^{t\wedge \bar{\tau}_{m}}\expo^{-\delta s}\left[\dfrac{1}{\lambda}(1-\expo^{-\lambda \nu_s})X_{s-}^*-C_{\sell}\nu_s\right]\diff N^{\gamma}_{s}\notag\\
 &\quad+\int_0^{t\wedge \bar{\tau}_{m}} \expo^{-\delta s}[\mathcal{L}v(X_s^*,Y_s^*)ds+\gamma G(X_{s-}^*,Y_{s-}^*,\nu_{s}^*;v)]\diff s,\notag
 \end{align}
 where $\bar{\tau}_{m}:=\inf\{t>0:X^{*}_{t}>m\}$. Now, from the construction of $v$, we know that it is the solution to \eqref{HJBin2}. Therefore, we have that
 	\[
\int_0^{t\wedge \bar{\tau}_{m}} \expo^{-\delta s}[\mathcal{L}v(X_s^*,Y_s^*)ds+\gamma G(X_{s-}^*,Y_{s-}^*,\nu_{s}^*;v)]\diff s=0.
 	\]
 Hence, we obtain that
 \begin{align*}
 \expo^{-\delta (t\wedge \bar{\tau}_{m})}v(X_{t\wedge \bar{\tau}_{m}}^*,Y_{t\wedge\bar{\tau}_{m}}^*)&=v(x,y)+M_{t\wedge \bar{\tau}_{m}}^*+H_{t\wedge \bar{\tau}_{m}}^*\notag\\
 &\quad-\int_0^{t\wedge \bar{\tau}_{m}}\expo^{-\delta s}\left[\dfrac{1}{\lambda}(1-\expo^{-\lambda \nu_s^{*}})X_{s-}^*-C_{\sell}\nu_s^*\right]\diff N^{\gamma}_{s}.
 \end{align*}
 Then, taking expectations in the previous identity it follows that
 \begin{align}\label{aux-5}
 v(x,y)&=\mathbb{E}_{x}\left[\expo^{-\delta (t\wedge \bar{\tau}_{m})}v(X_{t\wedge \tau_{m}}^*,Y_{t\wedge \bar{\tau}_{m}}^*)\right]\notag\\
 &\quad+\mathbb{E}_{x}\left[\int_0^{t\wedge \bar{\tau}_{m}}\expo^{-\delta s}\left[\dfrac{1}{\lambda}(1-\expo^{-\lambda \nu_s^{*}})X_{s-}^*-C_{\sell}\nu_s^*\right]\diff N^{\gamma}_{s}\right].
 \end{align}
Now, letting $m,t\rightarrow\infty$ in \eqref{aux-5}, using \eqref{HJbdes3}, \eqref{int_cond} and Monotone Convergence Theorem, we get
 \begin{align*}
 v(x,y)\leq \mathbb{E}_{x}\left[\int_0^{\infty}\expo^{-\delta s}\left[\dfrac{1}{\lambda}(1-\expo^{-\lambda \nu_s{*}})X_{s-}^*-C_{\sell}\nu_s^*\right]\diff N^{\gamma}_{s}\right]=J_{x,y}(\xi^{\sell*},0)\leq u(x,y),
 \end{align*}
 which implies the result. For the case when $\mu-\frac{1}{2}\sigma^2<0$, let us take $(\overline{X}^{*},\overline{Y}^{*})$ as the state process associated with the strategy  $\xi^{\sell*j}$ given by \eqref{epsilon-optimal}. We can check that $\xi^{\sell*j}$ has payoff
 \begin{align}\label{aux-7}
 J_{x,y}(\xi^{\sell*j},0)=\E_{x}\left[\int_0^j\expo^{-\delta t}[\overline{X}_t^*\circ_{\sell} \diff\xi^{\sell*}_t-C_{\sell}\diff \xi^{\sell*}_t]\right]+\frac{1}{\lambda}\E_{x}\left[\overline{X}^*_{\tilde{\tau}_j}\left[1-\expo^{-\lambda(y-\overline{Y}_{\tilde{\tau}_j}^*)}\right]\right],
 \end{align}
 where $\tilde{\tau}_j=\inf\{T_i>0: T_i>j\}$. Also, we see that \eqref{aux-5} is satisfied, if $t=j$, $(X^{*},Y^{*})$ is replaced by $(\overline{X}^{*},\overline{Y}^{*})$, and $\tau_m$ with  the stopping time $\bar{\tau}_{m}=\inf\{t>0:\overline{X}^{*}_{t}>m\}$. Then, letting $m\to\infty$ in \eqref{aux-5}, it follows
 \begin{align}\label{aux-6}
 \E_{x}\left[\int_0^j\expo^{-\delta t}[\overline{X}_t^*\circ_{\sell} \diff\xi^{\sell*}_t-C_{\sell}\diff\xi^{\sell*}_t]\right]=v(x,y)-\mathbb{E}_{x}\left[\expo^{-\delta j}v(\overline{X}_{j}^*,\overline{Y}_{j}^*)\right].
\end{align}
Now, applying \eqref{aux-6} in \eqref{aux-7}, 
\begin{align*}
J_{x,y}(\xi^{\sell*j},0)=v(x,y)-\mathbb{E}_{x}\left[\expo^{-\delta j}v(\overline{X}_{j}^*,\overline{Y}_{j}^*)\right]+\frac{1}{\lambda}\E_{\blue{x}}\left[\overline{X}^*_{\tilde{\tau}_j}\left[1-\expo^{-\lambda(y-\overline{Y}_{\tilde{\tau}_j}^*)}\right]\right].
\end{align*}
Therefore noting that the right-hand side of this expression converges to $v(x,y)$ as $j\to\infty$ allow us to establish that $\{\xi^{\sell*j}\}_{j=1}^{\infty}$ is a sequence of $\varepsilon$-optimal strategies.
\end{proof}
\begin{remark}
	Notice that  as a consequence  of Proposition \ref{Tcons1} and Lemma \ref{ac1},   we obtain that
	\begin{align*}
	\lim_{\gamma\to\infty}F_{\gamma}=\frac{nC_{\sell}}{n-1}:=F_{\infty},\quad
	\lim_{\gamma\to\infty}A_{\gamma}=0,\quad \text{and}\quad 
	\lim_{\gamma\to\infty}C_{\gamma}=\frac{F_{\infty}}{\lambda n^2}+\frac{1}{\lambda}\left(C_{\sell}\ln F_{\infty}-F_{\infty}\right).
	\end{align*}
	Recall that $F_{\gamma},C_{\gamma}, A_{\gamma}$ are given in \eqref{F1}--\eqref{Cg2}, respectively. Hence, straightforward computations yield
	\begin{equation*}
	\lim_{\gamma\to\infty}v(x,y)=
	\begin{cases}
	\displaystyle 0, &\text{if}\ y=0\ \text{and}\ x>0,\\
	\displaystyle\dfrac{(1-\expo^{-\lambda ny})x^n}{\lambda n^2 F_{\infty}^{n-1} },& \text{if}\ y>0\ \text{and}\ x<F_{\infty},\\
	\displaystyle \frac{F_{\infty}}{\lambda n^2}\left(1-\left(\frac{x\expo^{-\lambda y}}{F_{\infty}}\right)^n\right)\\
	\hspace{3cm}+\dfrac{x-F_{\infty}}{\lambda}-\dfrac{C_{\sell}}{\lambda}\ln \dfrac{x}{F_{\infty}} , &\text{if}\ y>0\ \text{and}\ F_{\infty}\leq x< F_{\infty}\expo^{\lambda y},\\
	\displaystyle \dfrac{ x(1-\expo^{-\lambda y})}{\lambda}-C_{\sell}y,&\text{if}\ y>0\ \text{and}\  x\geq F_{\infty}\expo^{\lambda y}.
	\end{cases}
	\end{equation*}
	These asymptotic limits allow us to  recover the value function for the case of singular strategies for the optimal execution problem obtained by Guo and Zervos in Proposition 5.1 of \cite{GS}.
\end{remark}

\renewcommand\thesection{A}
\setcounter{equation}{0}
\section*{Appendix. Proofs of some technical results}\label{techn1}
\begin{proof}[Proof of Lemma \ref{sell.1}]
	Let us take $(\xi^{\sell},\xi^{\barrier})\in\mathcal{A}(y)$ and $\Theta^{\sell,\barrier}=\{(\nu^{\sell}_{T_1},\nu^{\sell}_{T_2},\dots),(\nu^{\barrier}_{T_1},\nu^{\barrier}_{T_2},\dots)\}$ the selling-buying strategy associated with $(\xi^{\sell},\xi^{\barrier})$. Let $\mathcal{T}^{\sell}\subset\mathcal{T}$ the subset of decision times whose elements $\kappa_{i}$ are given in the following way 
	\begin{equation*}
		\begin{split}
			\kappa_{1}&\eqdef\inf\{T_{j}\in\mathcal{T}:\nu^{s}_{T_{j}}>0\},\\
			\kappa_{i}&\eqdef\inf\{T_{j}\in\mathcal{T}:T_{j}>\kappa_{i-1}\ \text{and}\ \nu^{s}_{T_{j}}>0\},\ \text{for}\ i\in\{2,3,\dots \}.\\
		\end{split}
	\end{equation*}
	From here, we see
	\begin{equation*}
		\begin{cases}
			\nu_{T_{j}}^{\barrier}=0<\nu_{T_{j}}^{\sell},&\text{if}\ T_{j}=\kappa_{i},\\
			\nu_{T_{j}}^{\sell}=0\leq\nu_{T_{j}}^{\barrier},&\text{if}\ \kappa_{i-1}<T_{j}<\kappa_{i},
		\end{cases}
	\end{equation*}
	for $i\in\{1,2,3,\dots \}$, with $\kappa_{0}=0$. Let $\alpha_{i,j}$ be defined as
	\begin{equation*}
		\alpha_{i}^{j}\eqdef\bigg(\displaystyle\sum_{\iota=i}^{j}\bigg(\nu_{\kappa_{\iota}}^{\sell}-\displaystyle\sum_{k=1}^{\infty}\nu^{\barrier}_{T_{k}}1_{\{\kappa_{\iota-1}<T_{k}\leq \kappa_{\iota}\}}\bigg)\bigg)^{+},\ \text{with}\ i,j\in\{1,2,\dots\}\ \text{and}\ i\leq j.
	\end{equation*}
	We construct $\bar{\nu}^{\sell}$ in the following way:
	\begin{enumerate}[(i)]
		\item If $T_{j}\notin\mathcal{T}^{\sell}$, then  $\bar{\nu}_{T_{j}}^{\sell}=0$.	
		\item If $T_{j}=\kappa_{1}$, then, $\bar{\nu}^{\sell}_{\kappa_{1}}\eqdef\alpha_{1,1}$. 
		
		\item If $T_{j}=\kappa_{2}$,
		\begin{equation*}
			\bar{\nu}^{\sell}_{\kappa_{2}}\eqdef
			\begin{cases}
				\alpha_{2}^{2},&\text{if}\ \bar{\nu}^{\sell}_{\kappa_{1}}>0,\\
				\alpha_{1}^{2},&\text{if}\ \bar{\nu}^{\sell}_{\kappa_{1}}=0.
			\end{cases}
		\end{equation*} 
		
		\item If $T_{j}=\kappa_{3}$,
		\begin{equation*}
			\bar{\nu}^{\sell}_{\kappa_{3}}\eqdef
			\begin{cases}
				\alpha_{3}^{3},&\text{if}\ \bar{\nu}^{\sell}_{\kappa_{1}}\geq0\ \text{and}\ \bar{\nu}^{\sell}_{\kappa_{2}}>0,\\
				\alpha_{2}^{3},&\text{if}\ \bar{\nu}^{\sell}_{1}>0\ \text{and}\ \bar{\nu}^{\sell}_{2}=0,\\
				\alpha_{1}^{3},&\text{if}\ \bar{\nu}^{\sell}_{\kappa_{\rho}}=0,\ \text{for}\ \rho\in\{1,2\}.
			\end{cases}
		\end{equation*}
		
		\item Recursively, if $T_{j}=\kappa_{i}$, with $i\in\{3,4,\dots\}$,
		\begin{equation*}
			\bar{\nu}^{\sell}_{\kappa_{i}}\eqdef
			\begin{cases}
				\alpha_{i}^{i},&\text{if}\ \bar{\nu}^{\sell}_{\kappa_{i-1}}>0\ \text{and}\ \bar{\nu}^{\sell}_{\kappa_{\rho}}\geq0\ \text{for}\ \rho\in\{1,\dots,i-2\}, \\
				\alpha_{i-1}^{i},&\text{if}\ \bar{\nu}^{\sell}_{\kappa_{i-1}}=0,\  \bar{\nu}^{\sell}_{\kappa_{i-2}}>0\ \text{and}\ \bar{\nu}^{\sell}_{\kappa_{\rho}}\geq0 \ \text{for}\ \rho\in\{1,\dots,i-3\},\\
				\quad\vdots&\\
				\alpha_{i-(p-2)}^{i},&\text{if}\ \bar{\nu}^{\sell}_{\kappa_{\rho}}=0\ \text{for}\ \rho\in\{i-(p-2),\dots,i-1\},\ \bar{\nu}^{\sell}_{\kappa_{i-(p-1)}}>0\  \text{and}\\
				&\bar{\nu}^{\sell}_{\kappa_{\rho}}\geq0 \ \text{for}\ \rho\in\{1,\dots,i-p\},\\
				\quad\vdots&\\
				\alpha_{2}^{i},&\text{if}\ \bar{\nu}^{\sell}_{\rho}=0\ \ \text{for}\ \rho\in\{2,\dots,i-1\},\ \text{and}\  \bar{\nu}^{\sell}_{1}>0,\\
				\alpha_{1}^{i},&\text{if}\ \bar{\nu}^{\sell}_{\kappa_{\rho}}=0,\ \text{with}\ \rho\in\{1,\dots,i-1\}.
			\end{cases}
		\end{equation*}
	\end{enumerate}
	Now, we take the selling strategy  $\bar{\xi}^{\sell}$ as
	\begin{equation}\label{s.2}	
	\bar{\xi}^{\sell}_{t}\eqdef\int_{0}^{t}\bar{\nu}^{\sell}_{s}\diff N^{\gamma}_{s}=\sum_{j=1}^{\infty}\bar{\nu}^{\sell}_{T_{j}}1_{\{T_{j}\leq t\}},\ \text{for}\ t\geq0,
	\end{equation}
	and define $\overline{Y}_{t}\eqdef y-\bar{\xi}^{\sell}_{t}$ for all $t\geq0$. Note that $\bar{\xi}^{\sell}_{t}=\xi^{\sell}_{t}=0$, when $0\leq t<\kappa_{1}$. Then, $\overline{Y}_{t}=y-\bar{\xi}^{\sell}_{t}\geq0$ and
	\begin{equation*}
		\xi^{\sell}_{t}-\xi^{\barrier}_{t}=-\displaystyle\sum_{j=1}^{\infty}\nu^{\barrier}_{T_{j}}1_{\{T_{j}\leq t\}}\leq 0=\bar{\xi}^{\sell}_{t}.
	\end{equation*} 
	Let us take $t\in[\kappa_{1},\kappa_{2})$. If $\bar{\nu}^{\sell}_{\kappa_{1}}>0$, it means that the agent sold $\bar{\nu}^{\sell}_{\kappa_{1}}$ shares which are a fraction (or the total) of his $y$ initial  shares. Otherwise, he only sold a fraction (or the total) of $\displaystyle\sum_{k=1}^{\infty}\nu^{\barrier}_{T_{k}}1_{\{0<T_{k}\leq \kappa_{1}\}}$ and $\displaystyle\sum_{k=1}^{\infty}\nu^{\barrier}_{T_{k}}1_{\{0<T_{k}\leq \kappa_{1}\}}-\nu^{\sell}_{\kappa_{1}}$ are accumulated for the next occasion when the agent decides to sell.  Then,  $\overline{Y}_{t}=y-\bar{\nu}^{\sell}_{\kappa_{1}}\geq0$. On the other hand,
	\begin{equation*}
		\begin{split}
			\xi^{\sell}_{t}&=\nu^{s}_{\kappa_{1}}\geq\bigg(\nu^{\sell}_{\kappa_{1}}-\sum_{k=1}^{\infty}\nu_{T_{k}}^{\barrier}1_{\{0<T_{k}\leq\kappa_{1}\}}\bigg)^{+}=\bar{\nu}^{\sell}_{\kappa_{1}}= \bar{\xi}^{\sell}_{t},\\
			\xi^{\sell}_{t}-\xi^{\barrier}_{t}&=\nu^{\sell}_{\kappa_{1}}-\sum_{k=1}^{\infty}\nu_{T_{k}}^{\barrier}1_{\{0<T_{k}\leq\kappa_{1}\}}-\sum_{k=1}^{\infty}\nu_{T_{k}}^{\barrier}1_{\{\kappa_{1}<T_{k}\leq t\}}\leq \bar{\nu}_{\kappa_{1}}=\bar{\xi}^{\sell}_{t}.
		\end{split}
	\end{equation*}
	Let us take $t\in[\kappa_{2},\kappa_{3})$. If $\bar{\nu}^{\sell}_{\kappa_{2}}>0$, at the same way that the above case, we have that the agent sold $\bar{\nu}^{\sell}_{\kappa_{2}}$ shares which are a fraction (or the total) of $y-\bar{\nu}_{\kappa_{1}}^{\sell}$. Otherwise, he only sold a fraction (or the total) of $\displaystyle\sum_{k=1}^{\infty}\nu^{\barrier}_{T_{k}}1_{\{\kappa_{1}<T_{k}\leq \kappa_{2}\}}$, or, a fraction of $\displaystyle\sum_{k=1}^{\infty}\nu^{\barrier}_{T_{k}}1_{\{0<T_{k}\leq \kappa_{2}\}}-\nu^{\sell}_{\kappa_{1}}$. Then, $\displaystyle\sum_{k=1}^{\infty}\nu^{\barrier}_{T_{k}}1_{\{\kappa_{1}<T_{k}\leq \kappa_{2}\}}-\nu^{\sell}_{\kappa_{2}}$, or, $\displaystyle\sum_{k=1}^{\infty}\nu^{\barrier}_{T_{k}}1_{\{0<T_{k}\leq \kappa_{2}\}}-(\nu^{\sell}_{\kappa_{1}}+\nu^{\sell}_{\kappa_{2}})$ are accumulated for the next occasion when the agent decides to sell.  Then,   $\overline{Y}_{t}=y-(\bar{\nu}^{\sell}_{\kappa_{1}}+\bar{\nu}^{\sell}_{\kappa_{2}})\geq0$. On the other hand,
	\begin{equation*}
		\begin{split}
			\xi^{\sell}_{t}&=\nu^{\sell}_{\kappa_{1}}+\nu^{\sell}_{\kappa_{2}}\geq\bar{\nu}^{\sell}_{\kappa_{1}}+\bar{\nu}^{\sell}_{\kappa_{2}}= \bar{\xi}^{\sell}_{t},\\
			\xi^{\sell}_{t}-\xi^{\barrier}_{t}&=\nu^{\sell}_{\kappa_{1}}-\sum_{k=1}^{\infty}\nu_{T_{k}}^{\barrier}1_{\{0<T_{k}\leq\kappa_{1}\}}\\
			&\quad+\nu^{\sell}_{\kappa_{2}}-\sum_{k=1}^{\infty}\nu_{T_{k}}^{\barrier}1_{\{\kappa_{1}<T_{k}\leq\kappa_{2}\}}-\sum_{k=1}^{\infty}\nu_{T_{k}}^{\barrier}1_{\{\kappa_{2}<T_{k}\leq t\}}\leq \bar{\nu}_{\kappa_{1}}+\bar{\nu}_{\kappa_{2}}=\bar{\xi}^{\sell}_{t}.
		\end{split}
	\end{equation*}
	Recursively, we can see that if $\kappa_{i-1}\leq t<\kappa_{i}$, with $i\in\{3,4,\dots\}$, then, $\overline{Y}_{t}=y-\displaystyle\sum_{\iota=1}^{i-1}\bar{\nu}^{\sell}_{\kappa_{\iota}}\geq0$ and 
	\begin{equation*}
		\begin{split}
			\xi^{\sell}_{t}&=\sum_{\iota=1}^{i-1}\nu^{s}_{\kappa_{\iota}}\geq\sum_{\iota=1}^{i-1}\bar{\nu}^{\sell}_{\kappa_{\iota}}= \bar{\xi}^{\sell}_{t},\\
			\xi^{\sell}_{t}-\xi^{\barrier}_{t}&=\sum_{\iota=1}^{i-1}\bigg(\nu^{\sell}_{\kappa_{\iota}}-\sum_{k=1}^{\infty}\nu_{T_{k}}^{\barrier}1_{\{\kappa_{\iota-1}<T_{k}\leq\kappa_{\iota}\}}\bigg)-\sum_{k=1}^{\infty}\nu_{T_{k}}^{\barrier}1_{\{\kappa_{i-1}<T_{k}\leq t\}}\leq\sum_{\iota=1}^{i-1}\bar{\nu}^{\sell}_{\kappa_{\iota}}=\bar{\xi}^{\sell}_{t}.
		\end{split}
	\end{equation*}
	Therefore, by the previously seen, we conclude $\bar{\xi}^{\sell}\in\mathcal{A}^{\sell}(y)$ and $\xi^{\sell}_{t}-\xi^{\barrier}_{t}\leq\bar{\xi}^{\sell}_{t}\leq\xi^{\sell}_{t}$ for all $t\geq0$.
\end{proof}

\begin{proof}[Proof of Proposition \ref{Tcons1}. (Construction of \eqref{HJBg1})]
	The proof of this result shall be given in two parts. In the first part, by smooth fit arguments, it is constructed the function $v$ which is a solution to the HJB equation \eqref{HJBin2}. In the second part, we prove that $v$ is in $\hol^{2,1}(\R_{+}\times\R_{+})$. Let $x<F_{\gamma}$ and consider Eq. \eqref{HJB0}.  In this case, the only solution that remains bounded as $x\downarrow 0$ is given by
	\begin{equation}\label{solv1}
	v(x,y)=A_1(y)x^{n},
	\end{equation}
	where $n$ is the positive solution to \eqref{pos1}. In order to find the form of the function $A_1(y)$ that appears above, we study the behaviour  of the solution  $v(x,y)=A_1(y)x^{n}$ along the boundary $x=F_{\gamma}$. Now we look for a solution that is continuous differentiable at the boundary $x=F_{\gamma}$. Evaluating \eqref{solv1} on the   left hand side of the equality in \eqref{HJBin5}, and recalling that $\mathbb{Y}(x)=\dfrac{1}{\lambda}\ln(x/F_{\gamma})$, we obtain
	\begin{align*}
		\mathcal{L}_{\gamma}&v(x,y)+\gamma\left[v(F_{\gamma},y-\mathbb{Y}(x))+\dfrac{1}{\lambda}(1-\expo^{-\lambda \mathbb{Y}(x)})x-C_{\sell}\mathbb{Y}(x)\right]\\
		&=-\gamma A_1(y)x^{n}+\gamma A_1(y-\mathbb{Y}(x))F_{\gamma}^{n}+\dfrac{\gamma(x-F_{\gamma})}{\lambda}-\gamma C_{\sell}\mathbb{Y}(x)\defeq K(x,y).
	\end{align*}
	By differentiating  with respect to $x$, we get that
	\begin{align}\label{partK1}
		K_{x}(x,y)&=-\gamma nA_1(y)x^{n-1}-\gamma A_1'(y-\mathbb{Y}(x))\dfrac{F_{\gamma}^{n}}{\lambda x}+\dfrac{\gamma}{\lambda}-\gamma\dfrac{C_{\sell}}{\lambda x}.
	\end{align}
	In order for the solution to be continuously differentiable at the boundary, we take $x=F_{\gamma}$ in \eqref{partK1}, and note that
	\begin{align*}
		-\gamma nA_1(y)F_{\gamma}^{n-1}-\gamma A_1'(y)\dfrac{F_{\gamma}^{n}}{\lambda F_{\gamma}}+\dfrac{\gamma}{\lambda}-\gamma\dfrac{C_{\sell}}{\lambda F_{\gamma}}=0,
	\end{align*}
	where the equality follows since \eqref{HJBin5} holds in $x=F_{\gamma}$. The above equation is equivalent to the following ODE for $A_1$,
	\begin{align*}
		A'_1(y)F_{\gamma}^{n}=-\lambda nA_1(y)F_{\gamma}^{n}+F_{\gamma}-C_{\sell}.
	\end{align*}
	The solution of this equation  is given by
	\begin{equation*}
		A_1(y)=\dfrac{(F_{\gamma}-C_{\sell})}{\lambda n F_{\gamma}^{n}}(1-\expo^{-\lambda ny}),
	\end{equation*}
	which implies that, when $x<F_{\gamma}$, the solution to the HJB equation \eqref{HJBin2}   is given by
	\begin{equation}\label{u_r1}
	v(x,y)=\dfrac{(F_{\gamma}-C_{\sell})}{\lambda n F_{\gamma}^{n}}(1-\expo^{-\lambda ny})x^{n}.
	\end{equation}
	Now we look for the solution of the HJB equation within the region $F\leq x<F_{\gamma}\expo^{\lambda y}$. Since
	\begin{equation*}
		v(F_{\gamma}-, y-\mathbb{Y}(x))=\dfrac{(F_{\gamma}-C_{\sell})}{\lambda n }\left(1-\dfrac{x^{n}}{F_{\gamma}^{n}}\expo^{-\lambda ny}\right),
	\end{equation*}
	Eq. \eqref{HJBin5} is given by
	\begin{align}\label{HJBin2reg}
		\mathcal{L}_{\gamma}v(x,y)+\gamma\left[\dfrac{F_{\gamma}-C_{\sell}}{\lambda n }\left(1-\dfrac{x^{n}}{F_{\gamma}^{n}}\expo^{-\lambda ny}\right)+\dfrac{x-F_{\gamma}}{\lambda}-C_{\sell}\mathbb{Y}(x)\right]=0.
	\end{align}
	In order to find the solution to this equation, we look first at the following set of equations,
	\begin{align*}
		\mathcal{L}_{\gamma}v_1(x,y)+\gamma\dfrac{(F_{\gamma}-C_{\sell})}{\lambda n }\left(1-\dfrac{x^{n}}{F_{\gamma}^{n}}\expo^{-\lambda ny}\right)&=0,\\
		\mathcal{L}_{\gamma}v_2(x,y)+\gamma\left[\dfrac{(x-F_{\gamma})}{\lambda}-C_{\sell}\mathbb{Y}(x)\right]&=0.
	\end{align*}
	The solutions to the previous equations are given, respectively, by
	\begin{align*}
		v_1(x,y)&=-\dfrac{(F_{\gamma}-C_{\sell})}{\lambda nF_{\gamma}^{n}}\expo^{-\lambda ny}x^{n}+\dfrac{\gamma(F_{\gamma}-C_{\sell})}{\lambda n(\delta+\gamma)},\\
		v_2(x,y)&=\dfrac{\gamma}{\lambda(\eta+\gamma)}x-\dfrac{\gamma C_{\sell}}{\lambda(\delta+\gamma)}\ln x+C,
	\end{align*}
	with $\eta=\delta-\mu$ and $C\eqdef\displaystyle \dfrac{\gamma}{\lambda(\delta+\gamma)}\biggl(\dfrac{bC_{\sell}}{\delta+\gamma}+C_{\sell}\ln F_{\gamma}-F\biggr)$. Hence, the solution to \eqref{HJBin2reg} that remains bounded for large values of $\gamma$ is given by
	\begin{align*}
		v(x,y)&=A_{2}(y)x^{m_{\gamma}}-\dfrac{(F_{\gamma}-C_{\sell})}{\lambda nF_{\gamma}^{n}}\expo^{-\lambda ny}x^{n}+\dfrac{\gamma}{\lambda(\eta+\gamma)}x-\dfrac{\gamma C_{\sell}}{\lambda(\delta+\gamma)}\ln x+\dfrac{\gamma(F_{\gamma}-C_{\sell})}{\lambda n(\delta+\gamma)}+C,
	\end{align*}
	for some function $A_{2}:\mathbb{R}_+\longrightarrow\mathbb{R}$. Recall that $m_{\gamma}$ is the negative solution to \eqref{neg1}. Since $u$ satisfies $u(F_{\gamma}-,y)=u(F_{\gamma}+,y)$, we conclude that
	for each $F_{\gamma}\leq x<F_{\gamma}\expo^{\lambda y}$, the solution $u$ to the equation \eqref{HJBin2},  has the following expression 	\begin{align}\label{eq1}
		v(x,y)&=A_{\gamma}\biggl(\dfrac{x}{F_{\gamma}}\biggr)^{m_{\gamma}}-\dfrac{(F_{\gamma}-C_{\sell})}{\lambda nF_{\gamma}^{n}}\expo^{-\lambda ny}x^{n}\\
		&\quad+\dfrac{\gamma}{\lambda(\eta+\gamma)}x-\dfrac{\gamma C_{\sell}}{\lambda(\delta+\gamma)}\ln x+\dfrac{\gamma(F_{\gamma}-C_{\sell})}{\lambda n(\delta+\gamma)}+C,\notag
	\end{align}
	where $A_{\gamma}$ is as in \eqref{Cg1}. Finally, in order to obtain the value of the optimal barrier $F_{\gamma}$, look for a solution $v$ such that $v_{x}$ is continuous at $x=F_{\gamma}$. Since $v_{x}(F_{\gamma}-,y)=v_{x}(F_{\gamma}+,y)$, we get
	\begin{align*}
		F_{\gamma}-C_{\sell}=\dfrac{F_{\gamma}m_{\gamma}}{\delta+\gamma}\biggr(\dfrac{\delta}{n}-\dfrac{\gamma\mu}{\eta+\gamma}\biggl)-\dfrac{C_{\sell}m_{\gamma}}{\delta+\gamma}\biggl(\dfrac{\delta}{n}+\dfrac{\gamma b}{\delta+\gamma}\biggr)+\dfrac{\gamma F_{\gamma}}{\eta+\gamma}-\dfrac{\gamma C_{\sell}}{\delta+\gamma}.
	\end{align*}
	This implies that $F_{\gamma}$ is given as in \eqref{F1}. 
	
	Now, let us find a general solution to \eqref{HJBin2} for the region $x\geq F_{\gamma}\expo^{\lambda y}$. We have that a particular solution to \eqref{HJBin4} is given by
	\begin{equation*}
		\dfrac{\gamma}{\lambda(\eta+\gamma)}x(1-\expo^{-\lambda y})-\dfrac{\gamma}{\delta+\gamma}C_{\sell}y.
	\end{equation*}
	Then, the solution to the equation \eqref{HJBin4} that remains bounded for large values of $\gamma$  is  given by
	\begin{equation}\label{solHJB2}
	v(x,y)=A_{3}(y)x^{m_{\gamma}}+\dfrac{\gamma}{\lambda(\eta+\gamma)}x(1-\expo^{-\lambda y})-\dfrac{\gamma}{\delta+\gamma}C_{\sell}y,
	\end{equation}
	for some function $A_3(y):\mathbb{R}_+\longrightarrow\mathbb{R}$.
	Finally, to find the expressions for the function $A_{3}(y)$  involved in \eqref{solHJB2}, we ask that $v$ is continuous at $x=F_{\gamma}\expo^{\lambda y}$. Then, since $v(F_{\gamma}\expo^{\lambda y}-,y)=v(F_{\gamma}\expo^{\lambda y}+,y)$, it is not difficult to check that
	\begin{equation*}
		A_3(y)=\dfrac{A_{\gamma}(1-\expo^{-\lambda m_{\gamma}y})}{F_{\gamma}^{m_{\gamma}}}.
	\end{equation*}
	Therefore, for each $ x\geq F_{\gamma}\expo^{\lambda y}$, the solution $u$ has the following expression
	\begin{equation}\label{solHJB3}
	v(x,y)=\dfrac{A_{\gamma}(1-\expo^{-\lambda m_{\gamma}y})x^{m_{\gamma}}}{F_{\gamma}^{m_{\gamma}}}+\dfrac{\gamma(1-\expo^{-\lambda y})x}{\lambda(\eta+\gamma)}-\dfrac{\gamma C_{\sell}y}{\delta+\gamma}.
	\end{equation}
	From \eqref{u_r1}, \eqref{eq1}, \eqref{solHJB3} and since $v(x,0)=0$, we conclude that the solution $v$ to the HJB equation \eqref{HJBin2} is given by \eqref{HJBg1}.
\end{proof}

Now, we shall proceed to verify that $v$, given in \eqref{HJBg1}, belongs to $\hol^{2,1}(\mathbbm{R}^{+}\times\mathbbm{R}^{+})$.
\begin{proof}[Proof of Proposition \ref{Tcons1}. (Smooth of \eqref{HJBg1})]
	Note that by construction, it is sufficient to show that  $v$ is $\hol^{2,1}$ at $x=F_{\gamma}$ and $x=F_{\gamma}\expo^{\lambda y}$, respectively, since $v\in\hol^{2,1}((\mathbbm{R}^{+}\times\mathbbm{R}^{+})\setminus\mathcal{A})$, where
	\begin{equation*}
		\mathcal{A}\eqdef\{(x,y)\in\mathbbm{R}^{+}\times\mathbbm{R}^{+}: x=F_{\gamma}\ \text{or}\ y=F_{\gamma}\expo^{\lambda y} \}.
	\end{equation*}
	
	\noindent\textit{Smooth fit at the variable $y$.} Using \eqref{HJBg1}, it is easy to see that $v_{y}(F_{\gamma}-,y)=v_{y}(F_{\gamma}+,y)$. Implying that $v_{y}$ is continuous at $x=F_{\gamma}$. Calculating first derivative in \eqref{HJBg1}, it can be verified that
	\begin{equation}\label{fit1}
	\begin{cases}
	v_y(F_{\gamma}\expo^{\lambda y}-,y)=F_{\gamma}-C_{\sell},\\
	v_y(F_{\gamma}\expo^{\lambda y}+,y)=\lambda m_{\gamma}A_{\gamma}+\dfrac{\gamma F_{\gamma}}{\eta+\gamma}-\dfrac{\gamma C_{\sell}}{\delta+\gamma}.
	\end{cases}
	\end{equation}
	From \eqref{F1}--\eqref{Cg1}, it can be verified that
	\begin{align}\label{fit3}
		\lambda m_{\gamma}A_{\gamma}+\dfrac{\gamma F_{\gamma}}{\eta+\gamma}-\dfrac{\gamma C_{\sell}}{\delta+\gamma}-(F_{\gamma}-C_{\sell})=0.
	\end{align}
	Then, by \eqref{fit1}--\eqref{fit3}, it yields that $v_y(F_{\gamma}\expo^{\lambda y}-,y)=F_{\gamma}-C_{\sell}=v_y(F_{\gamma}\expo^{\lambda y}+,y)$. Therefore $v_{y}$ is continuous at $x=F_{\gamma}\expo^{\lambda y}$.
	
	\noindent\textit{Smooth fit at the variable $x$.} We will show that $v_{xx}$ is continuous on $\mathbbm{R}^{+}\times\mathbbm{R}^{+}$.
	We will first verify that $v_{x}$ is continuous at $x=F_{\gamma}\expo^{\lambda y}$. From \eqref{HJBg1}, it follows that
	\begin{equation*}
		v_x(x,y)=
		\begin{cases}
			\displaystyle\dfrac{A_{\gamma}m_{\gamma}x^{m_{\gamma}-1}}{F_{\gamma}^{m_{\gamma}}}-\dfrac{(F_{\gamma}-C_{\sell})\expo^{-\lambda ny}x^{n-1}}{\lambda F_{\gamma}^{n}}\\
			\hspace{4cm}+\dfrac{\gamma}{\lambda(\eta+\gamma)}-\dfrac{\gamma C_{\sell}}{\lambda(\gamma+\delta)x},&\text{if}\  F_{\gamma}\leq x<F_{\gamma}\expo^{\lambda y},\\
			\displaystyle\dfrac{A_{\gamma}m_{\gamma}(1-\expo^{-\lambda m_{\gamma}y})x^{m_{\gamma}-1}}{F_{\gamma}^{m_{\gamma}}}+\dfrac{\gamma(1-\expo^{-\lambda y})}{\lambda(\eta+\gamma)},& \text{if}\ x\geq F_{\gamma}\expo^{\lambda y}.
		\end{cases}
	\end{equation*}
	Then, using \eqref{fit3}, we get that
	\begin{align*}
		v_x(F_{\gamma}\expo^{\lambda y}+,y)&=\dfrac{A_{\gamma}m_{\gamma}(1-\expo^{-\lambda m_{\gamma}y})\expo^{\lambda y(m_{\gamma}-1)}}{F_{\gamma}}+\dfrac{\gamma(1-\expo^{-\lambda y})}{\lambda(\eta+\gamma)}\notag\\
		&=\dfrac{A_{\gamma}m_{\gamma}\expo^{\lambda y(m_{\gamma}-1)}}{F_{\gamma}}-\dfrac{\expo^{-\lambda y}}{\lambda F_{\gamma}}\biggr(\lambda A_{\gamma}m_{\gamma}+\dfrac{\gamma F_{\gamma}}{\eta+\gamma}\biggl)+\dfrac{\gamma}{\lambda(\eta+\gamma)}\\
		&=\dfrac{A_{\gamma}m_{\gamma}\expo^{\lambda y(m_{\gamma}-1)}}{F_{\gamma}}-\dfrac{(F_{\gamma}-C_{\sell})\expo^{-\lambda y}}{\lambda F_{\gamma}}+\dfrac{\gamma}{\lambda(\eta+\gamma)}-\dfrac{\gamma C_{\sell} \expo^{-\lambda y}}{\lambda F_{\gamma}(\gamma+\delta)}\\
		&=v_{x}(F_{\gamma}\expo^{\lambda y}-,y).
	\end{align*}
	We now show that  $v_{xx}$ is continuous on $\mathbbm{R}^{+}\times\mathbbm{R}^{+}$ using the fact that $v_{x}$ is continuous on $\mathbbm{R}^{+}\times\mathbbm{R}^{+}$.
	Since $v_{x}$ is continuous at $x=F_{\gamma}$, from \eqref{HJB0} and  \eqref{HJBin5}, we have
	\begin{align}\label{sec1}
		v_{xx}(F_{\gamma}-,y)=\dfrac{\delta v(F_{\gamma},y)-\mu v_{x}(F_{\gamma},y)}{\dfrac{1}{2}\sigma^{2}F_{\gamma}^{2}}=v_{xx}(F_{\gamma}+,y).
	\end{align}
	If $x=F_{\gamma}\expo^{\lambda y}$, using \eqref{HJBin5} and \eqref{HJBin4}, it follows that
	\begin{align}\label{sec2}
		v_{xx}(F_{\gamma}\expo^{\lambda y}-,y)&=\dfrac{1}{\sigma^{2}F_{\gamma}^{2}}\biggr[2\biggr((\delta+\gamma)
		v(F_{\gamma},y)\\
		&\quad-\mu F_{\gamma}v_{x}(F_{\gamma},y)+\gamma\biggr(C_{\sell}y-\dfrac{F_{\gamma}}{\lambda}(1-\expo^{-\lambda y})\biggl)\biggl)\biggl]\notag\\
		&=v_{xx}(F_{\gamma}\expo^{\lambda y}+,y).\notag
	\end{align}
	Hence \eqref{sec1} and \eqref{sec2}, we conclude that $v_{xx}$ is continuous on $\mathbbm{R}^{+}\times\mathbbm{R}^{+}$.
\end{proof}

\begin{proof}[Proof of Lemma \ref{ac1}]
	First, recall that $m_{\gamma}$ was defined as the negative solution of \eqref{neg1}, and  observe that $m_{\gamma}\underset{\gamma\rightarrow0}{\longrightarrow}m_{0}$, where $m_{0}$ is the negative solution to \eqref{pos1}. Letting $\gamma\rightarrow0$ in \eqref{a1}, it is easy to see that $a_{\gamma}\underset{ \gamma\rightarrow0}{\longrightarrow}1$. On the other hand, letting $\gamma\rightarrow\infty$ in \eqref{a1}, it can be verified that
	\begin{align}\label{a1.0}
		\lim_{\gamma\rightarrow\infty}a_{\gamma}&
		=\lim_{\gamma\rightarrow\infty}\dfrac{\displaystyle\dfrac{\delta}{m_{\gamma}}-\biggl(\dfrac{\delta}{n}+\dfrac{\gamma b}{\delta+\gamma}\biggr)}{\displaystyle\dfrac{\eta (\delta+\gamma)}{(\eta+\gamma)m_{\gamma}}-\biggr(\dfrac{\delta}{n}-\dfrac{\gamma\mu}{\eta+\gamma}\biggl)}=\dfrac{\delta+nb}{\delta -\mu n}.
	\end{align}
	Since $n$ is the positive solution to \eqref{pos1}, it yields that
	\begin{equation}\label{a1.1}
	\delta+nb=\dfrac{1}{2}\sigma^{2}n^{2},\quad \text{and}\quad \delta -\mu n=\dfrac{1}{2}\sigma^{2}n(n-1).
	\end{equation}
	Therefore, from \eqref{a1.0}, \eqref{a1.1}, we conclude that $a_{\gamma}\rightarrow\dfrac{n}{n-1}$, if $\gamma\rightarrow\infty$. Now we shall prove that $\displaystyle 1< a_{\gamma}< \frac{n}{n-1}$, for all $\gamma>0$. In order to prove this result we first note that by \eqref{a1.1}, 
	\begin{equation}\label{aux-4-new.1}
	\begin{cases}
	\dfrac{\delta}{n}+\dfrac{\gamma b}{\delta+\gamma}=\dfrac{\gamma\sigma^2 n^2+2\delta^2}{2n(\delta+\gamma)}>0,\\
	\dfrac{\delta}{n}-\dfrac{\gamma \mu}{\eta+\gamma}=\dfrac{\gamma\sigma^2 n(n-1)+2\delta\eta}{2n(\eta+\gamma)}>0.
	\end{cases}
	\end{equation} 
	On the other hand, we have for each $\gamma>0$,
	\begin{align}\label{aux-4-new.2}
		\frac{\delta}{\delta+\gamma}&-\frac{\eta}{\eta+\gamma}-\frac{\gamma m_{\gamma}}{\delta+\gamma}\left(\frac{ b}{\delta+\gamma}+\frac{\mu}{\eta+\gamma}\right)\\
		&=\frac{\gamma\mu}{(\delta+\gamma)(\eta+\gamma)}-\frac{\gamma m_{\gamma}}{(\delta+\gamma)}\left(\frac{(\eta+\gamma)\sigma^2+2\mu^2 }{2(\delta+\gamma)(\eta+\gamma)}\right)>0,\notag
	\end{align}
	since $m_{\gamma}<0$, with $\gamma>0$. \eqref{aux-4-new.1}, \eqref{aux-4-new.2} imply 
	\begin{align}\label{aux-4-new.3}
		0<\frac{\eta}{\eta+\gamma}-\frac{m_{\gamma}}{\delta+\gamma}\left(\frac{\delta }{n}-\frac{\gamma\mu}{\eta+\gamma}\right)<\frac{\delta}{\delta+\gamma}-\frac{m_{\gamma}}{\delta+\gamma}\left(\frac{\delta }{n}+\frac{\gamma b}{\delta+\gamma}\right).
	\end{align}
	Therefore using \eqref{a1},  \eqref{aux-4-new.3}, it follows that $1< a_{\gamma}$. In order to prove the remaining inequality we just note that using \eqref{a1} it is enough to show that
	\begin{align}
		\frac{n\eta}{(n-1)(\eta+\gamma)}-\frac{\delta}{\delta+\gamma}-\frac{m_{\gamma}}{\delta+\gamma}\left(\frac{\delta}{n-1}-\frac{\delta}{n}-\frac{\gamma b}{\delta+\gamma}-\frac{\gamma\mu n}{(n-1)(\eta+\gamma)}\right)>0.\label{aux-4-new}
	\end{align}
	Note that
	\begin{align}\label{aux-2-new}
		\frac{n\eta}{(n-1)(\eta+\gamma)}-\frac{\delta}{\delta+\gamma}=\frac{\gamma(\delta-n\mu)+\delta\mu}{(n-1)(\delta+\gamma)(\eta+\gamma)}>0,
	\end{align}
	since $\delta-n\mu>0$. Similarly using \eqref{a1.1} and the fact that $\eta=\delta-\mu$ we obtain that
	\begin{align}\label{aux-3-new}
		\frac{\delta}{n-1}-\frac{\delta}{n}-&\frac{\gamma b}{\delta+\gamma}-\frac{\gamma\mu n}{(n-1)(\eta+\gamma)}\\
		&=\frac{\delta}{n(n-1)}-\frac{\gamma b(n-1)(\eta+\gamma)+\gamma\mu n(\delta+\gamma)}{(n-1)(\delta+\gamma)(\eta+\gamma)}\notag\\
		&=\frac{\delta(\delta+\gamma)(\eta+\gamma)-n^2\gamma\mu(\delta+\gamma)-\gamma b n(n-1)(\eta+\gamma)}{n(n-1)(\delta+\gamma)(\eta+\gamma)}\notag\\
		&=\frac{\delta(\delta+\gamma)(\eta+\gamma)-\delta n\gamma\mu+\frac{1}{2}\sigma^2 n^2\gamma\mu(n-1)-\delta\gamma(\eta+\gamma)}{n(n-1)(\delta+\gamma)(\eta+\gamma)}\notag\\
		&=\frac{\delta^2\eta+\frac{\sigma^2}{2}n(n-1)(\delta+\gamma\mu)}{n(n-1)(\delta+\gamma)(\eta+\gamma)}>0.\notag
	\end{align}
	Therefore using \eqref{aux-2-new}, \eqref{aux-3-new}, and the fact that $m_{\gamma}<0$ we obtain that \eqref{aux-4-new} holds and hence $\displaystyle a_{\gamma}< \frac{n}{n-1}$.
\end{proof}
\begin{proof}[Proof of Proposition \ref{promax1}]
	In order to show that \eqref{c1} holds, it is enough to prove that
	\begin{equation*}
		\begin{cases}
			G_{l}(x,y,l;v)<0,& \text{for}\  x<F_{\gamma},\\
			G_l(x,y,\mathbb{Y}(x);v)=0\ \text{and}\  G_{ll}(x,y,\mathbb{Y}(x);v)<0,& \text{for}\ F_{\gamma}\leq x< F_{\gamma}\expo^{\lambda y},\\
			G_{l}(x,y,l;v)>0,& \text{for}\ x\geq F_{\gamma}\expo^{\lambda y},
		\end{cases}
	\end{equation*}
	where $\mathbb{Y}(x)=\dfrac{1}{\lambda}\ln(x/F_{\gamma})$.
	
	\textit{Maximum on the first zone.} Let $x<F_{\gamma}$. Taking first derivatives in \eqref{HJBg1}  and evaluating at the point $(x\expo^{-\lambda l},y-l)$, we get that
	\begin{align*}
		v_{x}(x\expo^{-\lambda l},y-l)&=\dfrac{F_{\gamma}-C_{\sell}}{\lambda F_{\gamma}^{n}}(\expo^{-\lambda nl}-\expo^{-\lambda ny})x^{n-1}\expo^{\lambda l},\\
		v_{y}(x\expo^{-\lambda l},y-l)&=\dfrac{F_{\gamma}-C_{\sell}}{F_{\gamma}^{n}}\expo^{-\lambda n y}x^{n}.
	\end{align*}
	Then,
	\begin{align*}
		G_{l}(x,y,l;v)&=-\lambda v_{x}(x\expo^{-\lambda l},y-l)x\expo^{-\lambda l}-v_{y}(x\expo^{-\lambda l},y-l)+\expo^{-\lambda l}x-C_{\sell}\\
		&=-\dfrac{F_{\gamma}-C_{\sell}}{F_{\gamma}^{n}}x^{n}\expo^{-\lambda nl}+\expo^{-\lambda l}x-C_{\sell}.
	\end{align*}
	Note that the above  expression is negative if and only if  \begin{equation}\label{c4}
	\dfrac{\expo^{-\lambda l}x-C_{\sell}}{(x\expo^{-\lambda l})^{n}}<\dfrac{F_{\gamma}-C_{\sell}}{F_{\gamma}^{n}}.
	\end{equation}
	Taking the first derivative with respect to $x$ on the left hand side of \eqref{c4}, we have
	\begin{equation*}
		\dfrac{\partial}{\partial x}\biggl(\dfrac{\expo^{-\lambda l}x-C_{\sell}}{(x\expo^{-\lambda l})^{n}}\biggr)=\dfrac{n-1}{x^{n+1}}\left(\dfrac{nC_{\sell}}{n-1}\expo^{\lambda l}-x\right)\expo^{\lambda l(n-1)}.
	\end{equation*}
	By \eqref{F1} and Lemma \ref{ac1}, we know that   $x<F_{\gamma}\expo^{\lambda l}<\dfrac{nC_{\sell}}{n-1}\expo^{\lambda l}$. Then, $\dfrac{\partial}{\partial x}\biggl(\dfrac{\expo^{-\lambda l}x-C_{\sell}}{(x\expo^{-\lambda l})^{n}}\biggr)>0$, which implies that  $\dfrac{\expo^{-\lambda l}x-C_{\sell}}{(x\expo^{-\lambda l})^{n}}$ is non-decreasing with respect to  $x$. It yields that
	\begin{equation*}
		\dfrac{\expo^{-\lambda l}x-C_{\sell}}{(x\expo^{-\lambda l})^{n}}<\dfrac{\expo^{-\lambda l}F_{\gamma}-C_{\sell}}{(F_{\gamma}\expo^{-\lambda l})^{n}},\ \text{for each}\  x<F_{\gamma}.
	\end{equation*}
	Showing  that  $\dfrac{\expo^{-\lambda l}F_{\gamma}-C_{\sell}}{(F_{\gamma}\expo^{-\lambda l})^{n}}<\dfrac{F_{\gamma}-C_{\sell}}{F_{\gamma}^{n}}$ we obtain \eqref{c4}, which is  equivalent to see that
	\begin{equation}\label{eq6}
	(\expo^{-\lambda l}-\expo^{-\lambda nl})a_{\gamma}<1-\expo^{-\lambda nl},
	\end{equation}
	since $F_{\gamma}=C_{\sell}a_{\gamma}$. Taking  $l^{*}\eqdef\dfrac{\ln n}{\lambda (n-1)}$, it can be verified that
	\begin{equation*}
		(\expo^{-\lambda l^{*}}-\expo^{-\lambda nl^{*}})a_{\gamma}=\max_{l}\{(\expo^{-\lambda l}-\expo^{-\lambda nl})a_{\gamma}\}.
	\end{equation*}
	Since $a_{\gamma}<\dfrac{n}{n-1}$ and $(n+1)^{n}<n^{n}(n+1)$, with $n>1$, we get that
	\begin{align}\label{eq8}
		(\expo^{-\lambda l^{*}}-\expo^{-\lambda nl^{*}})a_{\gamma}=(n^{-\frac{1}{n-1}}-n^{-\frac{n}{n-1}})a_{\gamma}<1-n^{-\frac{n}{n-1}}=1-\expo^{-\lambda nl^{*}}.
	\end{align}
	This means that \eqref{eq6}  is satisfied for any $l>l^{*}$.  Now, if $l\leq l^{*}$,  we shall prove the statement \eqref{eq6} by contradiction. Suppose that there exists $0\neq l_{1}\leq l^{*}$ such that
	\begin{align}
		(\expo^{-\lambda l_{1}}-\expo^{-\lambda n l_{1}})a_{\gamma}&=1-\expo^{-\lambda nl_{1}},\label{eq7}\\
		(\expo^{-\lambda l}-\expo^{-\lambda nl})a_{\gamma}&\geq1-\expo^{-\lambda nl},\ \text{for each}\ l\leq l_{1}\leq l^{*}.\notag
	\end{align}
	Since
	\begin{equation*}
		(\expo^{-\lambda l}-\expo^{-\lambda nl})a_{\gamma}\leq(\expo^{-\lambda l_{1}}-\expo^{-\lambda nl_{1}})a_{\gamma}\leq(\expo^{-\lambda l^{*}}-\expo^{-\lambda nl^{*}})a_{\gamma},
	\end{equation*}
	we have that $(\expo^{-\lambda l^{*}}-\expo^{-\lambda nl^{*}})a_{\gamma}\geq1-\expo^{-\lambda n l^{*}}$,  which is a contradiction with \eqref{eq8}. If $l^{*}< l_{1}$ and satisfies that \eqref{eq7}, we have that
	\begin{align*}
		1-\expo^{-\lambda l^{*}n}<1-\expo^{-\lambda l_{1}n}=(\expo^{-\lambda l_{1}}-\expo^{-\lambda n l_{1}})a_{\gamma}< (\expo^{-\lambda l^{*}}-\expo^{-\lambda nl^{*}})a_{\gamma},
	\end{align*}
	which contradicts   \eqref{eq8}. Therefore, \eqref{eq6} is true for any $l$ and it yields  \eqref{c4}. We conclude that the maximum on the right hand side of \eqref{c1} is achieved at $l=0$  when $x<F_{\gamma}$.
	
	\textit{ Maximum on the second zone.} Let $F_{\gamma}\leq x< F_{\gamma}\expo^{\lambda y}$. Taking first derivatives of $v$ and evaluating $(F_{\gamma},y-\mathbb{Y}(x))$ in them, it follows that
	\begin{align*}
		-\lambda F_{\gamma} v_{x}(F_{\gamma},y-\mathbb{Y}(x))&=-\dfrac{\delta m_{\gamma}(F_{\gamma}-C_{\sell})}{n(\delta+\gamma)}+(F_{\gamma}-C_{\sell})\expo^{-\lambda n(y-\mathbb{Y}(x))}\\
		&\quad-\dfrac{\gamma F_{\gamma}}{\eta+\gamma}\biggr(1-\dfrac{\mu m_{\gamma}}{\delta+\gamma}\biggl)+\dfrac{\gamma C_{\sell}}{\delta+\gamma}\biggr(1+\dfrac{b m_{\gamma}}{\delta+\gamma} \biggl),\\
		-v_{y}(F_{\gamma},y-\mathbb{Y}(x))&=-(F_{\gamma}-C_{\sell})\expo^{-\lambda n(y-\mathbb{Y}(x))}.
	\end{align*}
	Then, recalling that $F_{\gamma}=C_{\sell}a_{\gamma}$, where $a_{\gamma}$ is given in  \eqref{a1}, we get that
	\begin{align*}
		G_l(x,y,\mathbb{Y}(x);v)=a_{\gamma}C_{\sell}\biggr(\dfrac{\eta}{\eta+\gamma}-\dfrac{m_{\gamma}}{\delta+\gamma}\biggr(\dfrac{\delta}{n}-\dfrac{\gamma\mu}{\eta+\gamma}\biggl)\biggl)-\dfrac{C_{\sell}}{\delta+\gamma}\biggr(\delta-m_{\gamma}\biggl(\dfrac{\delta}{n}+\dfrac{b \gamma}{\delta+\gamma}\biggr)\biggl)=0.
	\end{align*}
	Therefore, $l=\mathbb{Y}(x)$ is a critical point of $G(x,y,l;v)$; recall that the definition of $G$ is given in \eqref{maxop1}. To verify that $l=\mathbb{Y}(x)$ is a maximum of $G(x,y,l;v)$, we need to see that 
	\begin{equation*}
		G_{ll}(x,y,\mathbb{Y}(x);v)<0. 
	\end{equation*}
	Firstly, note that
	\begin{equation}\label{eq11.0}
	\lambda^{2}F_{\gamma}v_{x}(F_{\gamma},y-\mathbb{Y}(x))=\lambda^{2}A_{\gamma}m_{\gamma}-\lambda(F_{\gamma}-C_{\sell})\expo^{-\lambda n(y-\mathbb{Y}(x))}+\dfrac{\lambda\gamma F_{\gamma}}{\eta+\gamma}-\dfrac{\lambda\gamma C_{\sell}}{\delta+\gamma}.
	\end{equation}
	Now, taking the second derivatives of $v$ and evaluating $(F_{\gamma},y-\mathbb{Y}(x))$ in them, it follows that
	\begin{equation}\label{eq11}
	\begin{cases}
	\lambda^{2}F_{\gamma}^{2}v_{xx}(F_{\gamma},y-\mathbb{Y}(x))=\lambda^{2}A_{\gamma}m_{\gamma}(m_{\gamma}-1)\\
	\hspace{4.5cm}-\lambda(F_{\gamma}-C_{\sell})(n-1)\expo^{-\lambda n(y-\mathbb{Y}(x))}+\dfrac{\lambda\gamma C_{\sell}}{\delta+\gamma},\\
	2\lambda F_{\gamma} v_{xy}(F_{\gamma},y-\mathbb{Y}(x))=2\lambda n(F_{\gamma}-C_{\sell})\expo^{-\lambda n(y-\mathbb{Y}(x))},\\
	v_{yy}(F_{\gamma},y-\mathbb{Y}(x))=-\lambda n(F_{\gamma}-C_{\sell})\expo^{-\lambda n(y-\mathbb{Y}(x))}.
	\end{cases}
	\end{equation}
	By \eqref{eq11.0}--\eqref{eq11}, we get  that
	\begin{align}\label{eq12}
		G_{ll}(x,y,\mathbb{Y}(x);v)&=\lambda^{2}F_{\gamma}^{2}v_{xx}(F_{\gamma},y-\mathbb{Y}(x))+2\lambda F_{\gamma} v_{xy}(F_{\gamma},y-\mathbb{Y}(x))\\
		&\quad+\lambda^{2}F_{\gamma}v_{x}(F_{\gamma},y-\mathbb{Y}(x))+v_{yy}(F_{\gamma},y-\mathbb{Y}(x))-\lambda F_{\gamma}\notag\\
		&=\dfrac{C_{\sell}\lambda m_{\gamma}^{2}}{\delta+\gamma}\biggl(\dfrac{\delta a_{\gamma}}{n}-\dfrac{a_{\gamma}\gamma\mu}{\eta+\gamma}-\dfrac{\delta}{n}-\dfrac{b\gamma}{\delta+\gamma}-\dfrac{a_{\gamma}\eta (\delta+\gamma)}{m_{\gamma}^{2}(\eta+\gamma)}\biggr).\notag
	\end{align}
	To see that the above  expression is negative,  we need only to prove that
	\begin{align}\label{eq14}
		a_{\gamma}\biggr(\dfrac{\delta}{n}-\dfrac{\gamma\mu}{\eta+\gamma}-\dfrac{\eta (\delta+\gamma)}{m_{\gamma}^{2}(\eta+\gamma)}\biggl)-\biggr(\dfrac{\delta}{n}+\dfrac{b\gamma}{\delta+\gamma}\biggl)<0,
	\end{align}
	which is equivalent to see
	\begin{equation}\label{eq14.2}
	a_{\gamma}\biggr(\dfrac{\delta}{n}-\dfrac{\mu\gamma}{\eta+\gamma}\biggl)-\biggr(\dfrac{\delta}{n}+\dfrac{b\gamma}{\delta+\gamma}\biggl)<0,
	\end{equation}
	since $-a_{\gamma}\dfrac{\eta (\delta+\gamma)}{(\eta+\gamma)m_{\gamma}^{2}}<0$. Verifying that
	\begin{equation}\label{eq14.1}
	\delta\biggr(\dfrac{\delta}{n}-\dfrac{\mu\gamma}{\eta+\gamma}\biggl)<\dfrac{\eta(\delta+\gamma)}{\eta+\gamma}\biggr(\dfrac{\delta}{n}+\dfrac{b\gamma}{\delta+\gamma}\biggl),
	\end{equation}
	and recalling that $a_{\gamma}$ is given by \eqref{a1}, it yields \eqref{eq14.2}. We shall show \eqref{eq14.1}. Observe that  
	\begin{equation}\label{eq14.1.0.3}
	\delta\biggr(\dfrac{\delta}{n}-\dfrac{\mu\gamma}{\eta+\gamma}\biggl),
	\end{equation}
	is non-increasing with respect to $\gamma>0$ and
	\begin{equation}\label{eq14.1.0.1}
	\begin{cases}
	\displaystyle\delta\biggr(\dfrac{\delta}{n}-\dfrac{\gamma\mu}{\eta+\gamma}\biggl)\uparrow\dfrac{\delta^{2}}{n},&\text{when}\ \gamma\rightarrow0,\\
	\displaystyle\delta\biggr(\dfrac{\delta}{n}-\dfrac{\gamma\mu}{\eta+\gamma}\biggl)\downarrow\delta\biggl(\dfrac{\delta}{n}-\mu\biggr),&\text{when}\ \gamma\rightarrow\infty.
	\end{cases}
	\end{equation}
	If $b>0$, i.e. $\dfrac{\sigma^{2}}{2}> \mu$, then
	\begin{equation}
	\dfrac{\eta(\delta+\gamma)}{\eta+\gamma}\biggr(\dfrac{\delta}{n}+\dfrac{b\gamma}{\delta+\gamma}\biggl),\label{eq14.1.0.0}
	\end{equation}
	is non-decreasing with respect to $\gamma>0$ and
	\begin{equation}\label{eq14.0.0.4}
	\begin{cases}
	\displaystyle\dfrac{\eta(\delta+\gamma)}{\eta+\gamma}\biggr(\dfrac{\delta}{n}+\dfrac{b\gamma}{\delta+\gamma}\biggl)\downarrow\dfrac{\delta^{2}}{n},&\text{when}\ \gamma\rightarrow0,\\
	\displaystyle\dfrac{\eta(\delta+\gamma)}{\eta+\gamma}\biggr(\dfrac{\delta}{n}+\dfrac{b\gamma}{\delta+\gamma}\biggl)\uparrow\eta\biggr(\dfrac{\delta}{n}+b\biggl),&\text{when}\ \gamma\rightarrow\infty,
	\end{cases}
	\end{equation}
	From here and by \eqref{eq14.1.0.1},  it follows \eqref{eq14.1} and therefore we have that \eqref{eq12} is negative. If $b\leq 0$, i.e. $\dfrac{\sigma^{2}}{2}\leq \mu$, it can be verified that \eqref{eq14.1.0.0} is non-increasing with respect to $\gamma>0$ and   
	\begin{equation}\label{eq14.1.0.2}
	\begin{cases}
	\displaystyle\dfrac{\eta(\delta+\gamma)}{\eta+\gamma}\biggr(\dfrac{\delta}{n}+\dfrac{b\gamma}{\delta+\gamma}\biggl)\uparrow\dfrac{\delta^{2}}{n},&\text{when}\ \gamma\rightarrow0,\\
	\displaystyle\dfrac{\eta(\delta+\gamma)}{\eta+\gamma}\biggr(\dfrac{\delta}{n}+\dfrac{b\gamma}{\delta+\gamma}\biggl)\downarrow\eta\biggr(\dfrac{\delta}{n}+b\biggl),&\text{when}\ \gamma\rightarrow\infty.
	\end{cases}
	\end{equation}
	Defining the function $h(\gamma)$ as
	\begin{equation}\label{eq14.0.0.8}
	h(\gamma)\eqdef \delta\biggr(\dfrac{\delta}{n}-\dfrac{\mu\gamma}{\eta+\gamma}\biggl)- \dfrac{\eta(\delta+\gamma)}{\eta+\gamma}\biggr(\dfrac{\delta}{n}+\dfrac{b\gamma}{\delta+\gamma}\biggl),
	\end{equation}
	we can see that 
	\begin{equation}
	\begin{cases}
	h(\gamma)\rightarrow 0,&\text{when}\ \gamma\rightarrow 0,\\
	h(\gamma)\rightarrow \delta\biggl(\dfrac{\delta}{n}-\mu\biggr)-\eta\biggr(\dfrac{\delta}{n}+b\biggl),&\text{when}\ \gamma\rightarrow \infty,
	\end{cases}
	\end{equation}
	since \eqref{eq14.1.0.1}, \eqref{eq14.1.0.2} hold. In order to show \eqref{eq14.1}, it is enough to prove that $h(\gamma)$ is non-increasing and 
	\begin{equation}\label{eq14.0.0.7}
	\delta\biggl(\dfrac{\delta}{n}-\mu\biggr)-\eta\biggr(\dfrac{\delta}{n}+b\biggl)<0.
	\end{equation}
	Since $n$ is the positive solution to \eqref{pos1} and is bigger than one, it follows that
	\begin{equation}\label{eq14.0.0.5}
	\frac{\delta}{\mu}=n+\frac{\sigma^2}{2\mu}n(n-1)>n,
	\end{equation}
	this yields that $n\mu<\delta$. Then, applying this in \eqref{eq14.0.0.5}, we have
	\begin{equation}\label{eq14.0.0.9}
	\dfrac{\delta\mu}{n}=\mu^2+\frac{\sigma^2\mu}{2}(n-1)<\frac{1}{2}\sigma^{2}\delta-\mu\biggl(\dfrac{1}{2}\sigma^{2}-\mu\biggr),
	\end{equation}
	which implies \eqref{eq14.0.0.7}. Now, taking first derivative in \eqref{eq14.0.0.8}, we see that
	\begin{equation}
	h'(\gamma)=\frac{\eta}{(\eta+\gamma)^2}\left(-\delta\mu+\frac{\delta\mu}{n}-b\eta\right).
	\end{equation}
	Using \eqref{eq14.0.0.9}, it can be shown  $-\delta\mu+\dfrac{\delta\mu}{n}-b\eta<0$. This implies that $h(\gamma)$ is a negative non-increasing function. Therefore, it is true \eqref{eq14.2} and we have that \eqref{eq12} is negative. Thus the maximum at the right hand side of \eqref{c1} is achieved at $l=\mathbb{Y}(x)$, when $ F_{\gamma}\leq x< F_{\gamma}\expo^{\lambda y}$.
	
	\textit{Maximum on the third zone.} Let $x\geq F_{\gamma}\expo^{\lambda y}$. Taking the first derivatives of $v$ and evaluating $(x\expo^{-\lambda l},y-l)$ in them, it follows that
	\begin{align*}
		-\lambda v_{x}(x\expo^{-\lambda l},y-l)x\expo^{-\lambda l}&=-\dfrac{\lambda m_{\gamma} A_{\gamma}x^{m_{\gamma}}(\expo^{-\lambda m_{\gamma} l}-\expo^{-\lambda m_{\gamma}y})}{F_{\gamma}^{m_{\gamma}}}-\dfrac{\gamma x(\expo^{-\lambda l}-\expo^{-\lambda y})}{\eta+\gamma},\\
		-v_{y}(x\expo^{-\lambda l},y-l)&=-\dfrac{\lambda m_{\gamma}A_{\gamma}(x\expo^{-\lambda y})^{m_{\gamma}}}{F^{m_{\gamma}}_{\gamma}}-\dfrac{\gamma x\expo^{-\lambda y}}{\eta+\gamma}+\dfrac{\gamma C_{\sell}}{\delta+\gamma}.
	\end{align*}
	Then,
	\begin{align*}
		G_{l}(x,y,l;v)&=-\lambda v_{x}(x\expo^{-\lambda l},y-l)x\expo^{-\lambda l}-v_{y}(x\expo^{-\lambda l},y-l)+\expo^{-\lambda l}x-C_{\sell}\\
		&=-\dfrac{\lambda m_{\gamma} A_{\gamma}(x\expo^{-\lambda l})^{m_{\gamma}}}{F_{\gamma}^{m_{\gamma}}}+\dfrac{\eta x\expo^{-\lambda l}}{\eta+\gamma}-\dfrac{\delta C_{\sell}}{\delta+\gamma}.
	\end{align*}
	To see that the above  expression is positive, is equivalent to show that
	\begin{align}\label{eq18}
		\dfrac{\eta x\expo^{-\lambda l}}{\eta+\gamma}- \dfrac{m_{\gamma}\lambda A_{\gamma}(x\expo^{-\lambda l})^{m_{\gamma}}}{F_{\gamma}^{m_{\gamma}}}>\dfrac{\delta C_{\sell}}{\delta+\gamma}.
	\end{align}
	Observe that from \eqref{Cg1} and \eqref{eq14.2}, it can be verified that $A_{\gamma}<0$. Then,  it follows that
	\begin{align*}
		\dfrac{\eta x\expo^{-\lambda l}}{\eta+\gamma}&\geq \dfrac{\eta F_{\gamma}\expo^{\lambda(y-l)}}{\eta+\gamma},\\
		- \dfrac{m_{\gamma}\lambda A_{\gamma}(x\expo^{-\lambda l})^{m_{\gamma}}}{F_{\gamma}^{m_{\gamma}}}&\geq - m_{\gamma}\lambda A_{\gamma}\expo^{\lambda m_{\gamma}(y-l)},
	\end{align*}
	since $x\geq F_{\gamma}\expo^{\lambda y}$  and $m_{\gamma}<0$. Then
	\begin{align}\label{eq18.0}
		\dfrac{\eta x\expo^{-\lambda l}}{\eta+\gamma}- \dfrac{m_{\gamma}\lambda A_{\gamma}(x\expo^{-\lambda l})^{m_{\gamma}}}{F_{\gamma}^{m_{\gamma}}}>\dfrac{\eta F_{\gamma}\expo^{\lambda(y-l)}}{\eta+\gamma}- m_{\gamma}\lambda A_{\gamma}\expo^{\lambda m_{\gamma}(y- l)}\defeq g(l).
	\end{align}
	Note that $g(l)$ is non-increasing with respect to $l$ and from \eqref{F1}--\eqref{Cg1}, we get that $g(y)=\dfrac{\delta C_{\sell}}{\delta+\gamma}$. Therefore, \eqref{eq18.0} yields that \eqref{eq18}. Thus, the maximum at the right hand side of \eqref{c1} is achieved at $l=y$, when $x\geq F_{\gamma}\expo^{\lambda y}$.
\end{proof}

\section*{Acknowledgements}
The research of D. Hern\'andez-Hern\'andez was partially supported by  CONACyT, under grant 254166.
H. A. Moreno-Franco acknowledges financial support from CIMAT, CONACyT and HSE. The last one  has been funded by the Russian Academic Excellence Project ``5-100''.

\end{document}